\documentclass[12pt]{amsart}

\usepackage{color}

\usepackage{graphicx}
\usepackage{epsfig}

\newtheorem{thm}{Theorem}[section]

\newtheorem{remark}[thm]{Remark}
\newtheorem{example}[thm]{Example}

\usepackage{amsmath}
\usepackage{amsxtra}
\usepackage{amscd}
\usepackage{amsthm}
\usepackage{amsfonts}
\usepackage{amssymb}
\usepackage{eucal}
\newcommand{\bmb}{\left( \begin{array}{rr}}
\newcommand{\enm}{\end{array}\right)}

\newcommand{\Z}{{\mathbb Z}}

\newcommand{\R}{{\mathbb R}}
\newcommand{\N}{{\mathbb N}}

\newcommand{\bv}{{\mathbf v}}
\newcommand{\bw}{{\mathbf w}}

\newcommand{\al}{{\alpha}}

\newcommand{\beq}{\begin{equation}}
\newcommand{\eeq}{\end{equation}}
\newcommand{\beqa}{\begin{eqnarray}}
\newcommand{\eeqa}{\end{eqnarray}}

\numberwithin{equation}{section}

\begin{document}

\title[Arctic curves for paths with arbitrary fixed starting points]{Arctic curves for paths with arbitrary starting points: a Tangent Method approach}
\author{Philippe Di Francesco} 
\address{
Institut de physique th\'eorique, Universit\'e Paris Saclay, 
CEA, CNRS, F-91191 Gif-sur-Yvette, FRANCE,
and 
Department of Mathematics, University of Illinois, Urbana, IL 61821, U.S.A. \hfill
\break  e-mail: philippe@illinois.edu
}
\author{Emmanuel Guitter}
\address{
Institut de physique th\'eorique, Universit\'e Paris Saclay, 
CEA, CNRS, F-91191 Gif-sur-Yvette, FRANCE. \hfill
\break  e-mail: emmanuel.guitter@ipht.fr
}

\begin{abstract}

We use the tangent method to investigate the arctic curve in a model of non-intersecting lattice paths with arbitrary fixed starting points 
aligned along some boundary and whose distribution is characterized by some arbitrary piecewise differentiable function. 
We find that the arctic curve has a simple explicit parametric representation depending of this function, providing us with a simple 
transform that maps the arbitrary boundary condition to the arctic curve location. We discuss generic starting point distributions as well 
as particular freezing ones which create additional frozen domains adjacent to the boundary, hence new portions for the arctic curve. 
A number of examples are presented, corresponding to both generic and freezing distributions.  
 
\end{abstract}

\maketitle
\date{\today}
\tableofcontents


\section{Introduction}
\label{introsec}
Many tiling problems of finite plane domains of large size are known to exhibit the so-called {\it arctic curve phenomenon}, namely the 
existence of a sharp phase separation between ``crystalline" (i.e. regularly tiled) phases often induced by boundary corners
and ``liquid" (i.e. disordered) phases away from the influence of boundaries. For instance, the celebrated problem of tiling 
the Aztec diamond with dominoes is known to display an arctic circle separating frozen phases induced by the corners of the domain
from an entropic phase 
in the center \cite{CEP,JPS}. Typically, one studies the asymptotics of tilings of scaled domains whose limits
are polygons. More generally, dimer models on regular graphs, which are a dual version of tiling problems,
exhibit the same arctic phenomenon, which received a fairly general treatment in the recent years \cite{KO1,KO2,KOS}. 
Free boundary conditions, where portions of the boundary are allowed to fluctuate were also 
studied \cite{DFR}.

The general method to obtain the arctic curve location is the asymptotic study of bulk expectation values, which requires a certain amount 
of technology, resorting for instance to the machinery of the Kasteleyn matrix. Other rigorous methods use the machinery of cluster integrable 
systems of dimers \cite{DFSG,KP}.

All the models above have an interesting common feature: they can be rephrased in terms of configurations of non-intersecting lattice (or graph)
paths, which arise from conservation laws of the models, and display their underlying fermionic character.  
Typically, we have a set of paths with steps along oriented edges of a regular graph, with fixed starting and ending points,
and subject to the condition that no two paths share the same vertex. These occupy a maximal domain $D$, which is then scaled to
reach a continuum limit. In the path formulation,
frozen phases correspond to regular compact configurations (such as zones with parallel paths only), or to empty domains not visited 
by any path. With such an interpretation, it is easy to track down the arctic curve (or portions thereof) as the asymptotic ``outer shell" of the
path configurations, determined by the outermost paths.
Inspired by this remark, Colomo and Sportiello \cite{COSPO} recently devised a new method for determining the arctic curve in path models,
coined the {\it tangent method}.  The idea is to move the endpoint of one of the outermost paths to some distant point $p$ on the regular graph, 
so as to force this path to exit the domain $D$ say at a point $\ell$. It is then argued that between $\ell$ and $p$, away from the influence 
of the other paths the most likely asymptotic trajectory is a straight line. Inside the domain $D$, the outermost path is expected to first follow
the outer shell, then escape this shell \emph{tangentially} and continue on towards $\ell$, again along a straight line since the crossed region 
is empty from other paths. For any fixed $p$, the most likely position $\ell=\ell(p)$ corresponds to having both straight lines identical. Solving
the corresponding  extremization problem therefore provides a parametric family of straight lines $(\ell(p)\, p)$, all tangent to the arctic curve,
which is then recovered as the envelope of this family of tangents.
The main advantage of this method is that it only requires to estimate a {\it boundary} one-point function, namely that for which the endpoint
of an outer path is moved to a position $\ell$ on the boundary of $D$. Such a function is considerably simpler to compute than bulk expectation 
values.

The method, though non-rigorous, was successfully  tested in a
number of examples \cite{COSPO,DFLAP}. Remarkably, it seems to even apply to situations where the lattice paths interact, 
such as the so-called osculating paths describing configurations of the six-vertex model. In this model, the path configurations are
allowed to form ``kissing points" where a vertex is shared by two neighboring paths. The tangent method predicts in particular the asymptotic shape of large 
alternating sign matrices (ASM) \cite{COSPO} as well as vertically symmetric alternating sign matrices (VSASM) \cite{DFLAP}.

In the present paper, we use the tangent method to investigate
path/tiling models with new kinds of boundary conditions: in the path language, we consider path configurations where the starting points
of the paths take fixed but \emph{arbitrary positions} aligned along some boundary segment. Asymptotically, the distribution of these points 
is simply characterized 
by some arbitrary piecewise differentiable function $\al(u):[0,1]\to \R$. 
Our main result is that the corresponding arctic curve has an explicit parametric representation $(X(t),Y(t))$ for its coordinates in the plane,
 which takes the following simple form:
\begin{equation}\label{mainresult}
\begin{split}
&\left\{ \begin{matrix}
X(t)= & t- \frac{\displaystyle{x(t)(1-x(t))}}{\displaystyle{x'(t)}} \hfill\\
& \\
Y(t)= &\frac{\displaystyle{(1-x(t))^2}}{\displaystyle{x'(t)}}\hfill
\end{matrix}\right. \\
&\qquad \hbox{with}\ \ x(t):= e^{\textstyle{-\int_0^1 \frac{du}{t-\al(u)}}}\ .\\
\end{split}
\end{equation}
This provides us with a direct transform that maps the ``boundary shape" $\al(u)$ to the arctic curve, made in general of
several portions corresponding to various allowed domains of the parameter $t$.  

The paper is organized as follows. 
In Section \ref{model}, we present the general path model that we will consider, together with its tiling interpretation, and compute its partition function. The model involves paths on the edges of the square lattice with
starting points fixed at arbitrary positions along a horizontal segment. As just mentioned, these positions are entirely characterized asymptotically by 
their limiting boundary shape $\al(u)$.
The tiling interpretation allows to rephrase the model in three different (but equivalent) ways, using different sets of paths.

The tangent method is then applied in Sections \ref{puzzleone} and \ref{puzzle2} using two different sets of paths 
to obtain two different portions of the arctic curve. The derivation involves the computation of a boundary 
one-point function, which is performed by using the LU decomposition of the 
Lindstr\"om-Gessel-Viennot matrix, a method advertised and successfully used in \cite{DFLAP} for similar problems. 
Both computations lead to the {\it same} parametric equations for the arctic curve, as given above, in two different parameter domains.

Section \ref{examplesec} presents various examples: the ``pure" case $\al(u)=p\, u$, the case of a piecewise linear $\al(u)$
and finally two instances of some non-linear $\al(u)$. 
Subtleties arise whenever $\al'(u)=1$ on finite segments, corresponding to a certain type of freezing boundary condition
inducing new macroscopic frozen regions inside the path domain. Likewise, macroscopic gaps 
in the distribution of starting points induce another type of freezing. 
These ``freezing boundaries" are investigated in detail in Section \ref{freezesec}, and give rise to additional portions of the arctic curve, still described 
by the parametric equations \eqref{mainresult} above, but for yet other domains of $t$.

We gather a few concluding remarks in Section~\ref{conclusec}.

\section{Definition of the model and partition function}
\label{model}

\subsection{Non-intersecting lattice paths with arbitrary starting points}
\label{NILP}

In its simplest formulation, the model that we wish to study simply describes configurations of non-intersecting lattice paths (NILP)
with prescribed extremities. More precisely, a configuration consists of $n+1$ lattice paths making only west- or north-oriented unit steps along
the edges of the regular square lattice, with respective starting points $O_i$ and endpoints $E_i$, $i=0,\dots, n$, chosen as follows: 
the endpoints $E_i$ are taken with coordinates $(0,i)$ so as to span a vertical segment of length $n$; 
the starting points $O_i$ have coordinates $(a_i,0)$ where $(a_i)_{0\leq i\leq n}$ is a given \emph{arbitrary strictly increasing sequence of integers}
of length $n+1$  with $a_0=0$. These vertices therefore lie on a horizontal segment of length $a_n$ (with $a_n \geq n$)  with prescribed but arbitrary
strictly increasing positions along this segment. The paths are required to be non-intersecting in the sense that any two paths cannot share a common vertex of the lattice.
Figure \ref{fig:path1} shows an example of such path configuration with $n=6$. Note that, due to the non-intersection constraint, the portions of
the paths lying above the line $Y=X$ in the $(X,Y)$ plane (dashed line in the figure) are "frozen" as they necessarily form horizontal segments.
 
\begin{figure}
\begin{center}
\includegraphics[width=8cm]{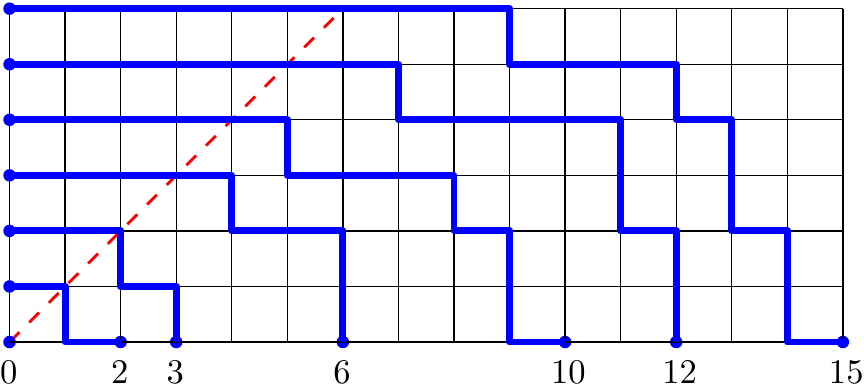}
\end{center}
\caption{A configuration of $n+1=7$ non-intersecting lattice paths made of west- or north-oriented unit steps, with starting points $O_i=(a_i,0)$ and endpoints $E_i=(0,i)$, $i=0,\dots, n$, here in the particular case $(a_i)_{0\leq i\leq n}=(0,2,3,6,10,12,15)$.  The portions of
paths above the dashed line are necessarily "frozen" into horizontal segments.}
\label{fig:path1}
\end{figure}

\subsection{Tiling interpretation and alternative path formulations}
\label{tiling}
As displayed in Figure \ref{fig:tiling123}, any of the above defined configurations of non-intersecting lattice paths may be transformed into a particular
\emph{tiling} for the domain of the plane covered by the paths. More precisely, to each horizontal edge $(p+1,q)\to (p,q)$ carrying a west-oriented step is associated an \emph{upper tile} which is the rhomboid with vertices $(p-1/2,q-1/2),(p+1/2,q-1/2),(p+3/2,q+1/2),(p+1/2,q+1/2)$, to each vertical edge 
$(p,q)\to (p,q+1)$ carrying a north-oriented step is associated a \emph{right tile} which is a rhomboid with vertices $(p-1/2,q-1/2),(p-1/2,q+1/2),(p+1/2,q+3/2),(p+1/2,q+1/2)$ and finally, to each unvisited vertex $(p,q)$ is associated a \emph{front tile} which is a square with vertices $(p-1/2,q-1/2),(p-1/2,q+1/2),(p+1/2,q+1/2),(p+1/2,q-1/2)$. Apart from the original NILP configuration, the resulting tiling naturally gives rise to two other sets of NILP as displayed in the figure.

\begin{figure}
\begin{center}
\includegraphics[width=10cm]{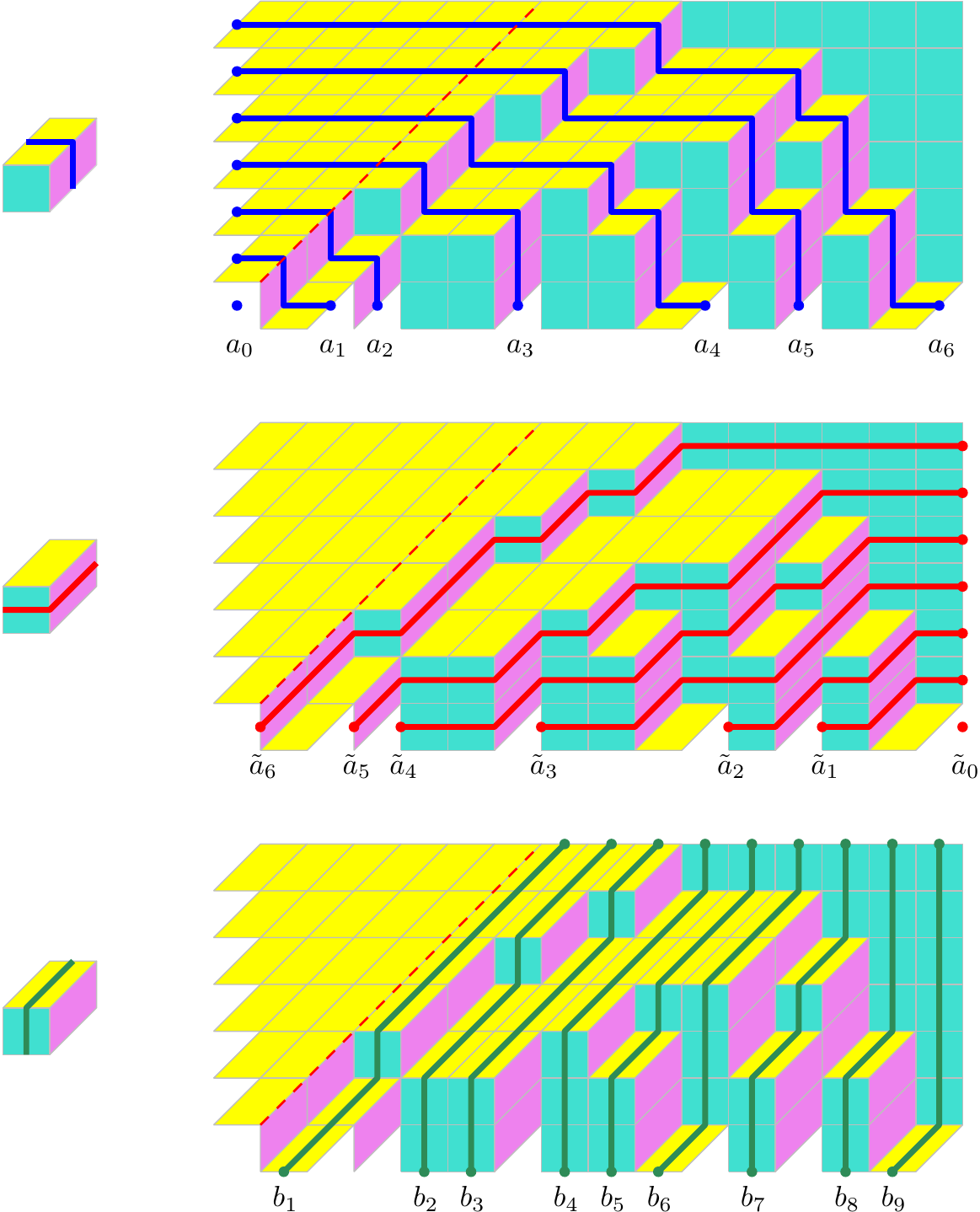}
\end{center}
\caption{Top: the path configuration of Figure \ref{fig:path1} together with the associated tiling configuration, made of upper, right, and front tiles, corresponding respectively to edges carrying a horizontal path step, edges carrying a vertical path step, and vertices not visited by the paths. Middle: 
connecting the vertical sides of both front tiles and right tiles by elementary segments creates a configuration of $(n+1)=7$ non-intersecting lattice paths 
made of east- and northeast-oriented elementary steps. These paths (numbered from right to left) have starting points $\tilde{O}_i=(a_n+1/2-\tilde{a}_i,0)$, where $\tilde{a}_i:=a_n-a_{n-i}$, and endpoints $\tilde{E}_i=(a_n+1/2,i)$, $i=0,\dots, n$. Here, $a_n=15$ and $(\tilde{a}_i)_{0\leq i\leq n}=(0,3,5,9,12,13,15)$. Bottom: connecting the horizontal sides of both front tiles and upper tiles (except for those in the frozen region above the dashed line) by elementary segments creates a configuration of $m=a_n-n=9$ non-intersecting lattice paths 
made of north- and northeast-oriented elementary steps. These paths have starting points $\hat{O}_i=(b_i,-1/2)$ and endpoints $\hat{E}_i=(n+i,n+1/2)$, $i=1,\dots, m$, where $(b_i)_{1\leq j\leq m}$ is the ``complementary sequence" of $(a_i)_{0\leq j \leq n}$ in $\Z\cap [0,a_n]$. Here, $(b_i)_{1\leq i\leq m}=(1,4,5,7,8,9,11,13,14)$.}
\label{fig:tiling123}
\end{figure}

The second set of paths is obtained by associating to the right and front tiles introduced above northeast- and east-oriented steps of the form $(p-1/2,q)\to(p+1/2,q+1)$, and $(p-1/2,q)\to(p+1/2,q)$ respectively. This leads to a configuration of $(n+1)$ NILP with endpoints
$\tilde{E}_i$ of coordinates $(a_n+1/2,i)$ and starting points $\tilde{O}_i$ of coordinates $(a_n+1/2-\tilde{a}_i,0)$ for $i=0,\dots,n$, where
$(\tilde{a}_i)_{0\leq i\leq n}$ is the strictly increasing sequence (with $\tilde{a}_0=0$) defined as:
\begin{equation}
\tilde{a}_i:=a_n-a_{n-i}\ .
\label{eq:tildeaidef}
\end{equation} 
\medskip

As for the the third set of paths, it is obtained by associating to the upper and front tiles mentioned above northeast- and north-oriented steps of the form $(p,q-1/2)\to(p+1,q+1/2)$, and $(p,q-1/2)\to(p,q+1/2)$ respectively. We omit here those upper tiles above the $Y=X$ line as they form a 
regular crystalline pattern and the associated paths play no role. This leads to a configuration of $m=a_n-n$ NILP with endpoints 
$\hat{E}_i$ of coordinates $(n+i,n+1/2)$ and starting points $\hat{O}_i$ of coordinates $(b_i,-1/2)$ for $i=1,\dots,m$, where the 
strictly increasing sequence $(b_i)_{1\leq i\leq m}$ is the \emph{complementary sequence} of the sequence $(a_i)_{0\leq i\leq n}$,
defined for instance via the polynomial identity
\begin{equation}
\prod_{i=1}^m(t-b_i)\, \prod_{i=0}^n(t-a_i)=\prod_{i=0}^{a_n}(t-i)\ , \quad \hbox{with}\ m=a_n-n\ .
\label{eq:bidef}
\end{equation} 
\medskip
 
Clearly, the data of any of the three path configurations allows to recover the two others so that each of the three descriptions
carries all the information about the configuration at hand. We may therefore use any of the three path formulations to describe
our model.

\subsection{Partition function}
\label{pf}
Returning to the original formulation of Section \ref{NILP} with paths made or west- and north-oriented steps, the \emph{partition function 
$Z_n := Z\left((a_i)_{0\leq i\leq n}\right)$}
of the model, namely the \emph{number} of non-intersecting path configurations, may be obtained via the famous   
Lindstr\"om-Gessel-Viennot (LGV) lemma \cite{LGV1,GV}, which states that $Z_n=\det\left((A_{i,j})_{0\leq i,j\leq n}\right)$ 
where $A_{i,j}$ denotes the number
of paths made of west- and north-oriented steps along edges of the square lattice and connecting the starting point $O_i$ to the endpoint $E_j$.
In the present case, we have clearly
\begin{equation*}
A_{i,j}={a_i+j\choose j}
\end{equation*}
since a path from $O_i$ to $E_j$ is made of a total of $a_i+j$ steps among which exactly $j$ are oriented north. This latter determinant
may be easily computed in various ways. We present here a derivation using the so-called LU decomposition of the matrix $A$ with elements 
$A_{i,j}$ above. This method will indeed prove adapted when we will extend our calculation to some more involved determinants with the same flavor
and was successfully applied for determining the arctic curve for various path problems in \cite{DFLAP}.
Recall that the LU decomposition consists in writing  the square matrix $A$, of size $(n+1)\times (n+1)$,  as the product $A=L\, U$ of a lower triangular  square matrix $L$ by an upper triangular square matrix $U$ (both matrices having the same size as $A$). Such a decomposition exists
for suitable matrices (among which is the desired matrix $A$, as made explicit below) and is moreover unique if we demand that $L$ is lower uni-triangular, i.e.\ $L_{i,i}=1$ for all $i=0,\dots,n$. From the knowledge of the matrices $L$ and $U$, we immediately obtain $Z_n$ via
\begin{equation*}
Z_n=\det(A)=\det(L)\times\det(U)=\prod_{i=0}^{n}U_{i,i}
\end{equation*} 
since $U$ is upper triangular and $\det(L)=1$. Note that, in practice, only the knowledge of the diagonal 
elements of $U$ is required to get $Z_n$.

In order to get the LU decomposition of the matrix $A$, it is enough to find a lower triangular square matrix 
$L^{-1}$ with diagonal elements equal to $1$ such that 
$L^{-1}\, A$ is upper triangular. We have the following result:
\begin{thm}
The lower uni-triangular matrix $L^{-1}$ with 
matrix elements
\begin{equation}
L^{-1}_{i,j}=\left\{
\begin{matrix}\frac{\displaystyle{\prod\limits_{s=0}^{i-1}(a_i-a_s)}}{\displaystyle{\prod\limits_{s=0\atop s\neq j}^{i}(a_j-a_s)}} & \hbox{for}\ i\geq j\\
0 & \hbox{for}\ i<j
\end{matrix}
\right.
\label{eq:Lmat}
\end{equation}
is such that $U:=L^{-1}\, A$ is upper triangular. 
\end{thm}
\begin{proof}
The diagonal elements of $L^{-1}$ are clearly equal to $1$ and, for any $i$ and $j$, we may write
\begin{equation*}
L^{-1}_{i,j}=\prod\limits_{s=0}^{i-1}(a_i-a_s) 
\oint_{\mathcal{C}(a_j )} \frac{1}{\displaystyle{\prod\limits_{s=0}^{i}(t-a_s)}}\frac{dt}{2{\rm i} \pi}\ ,
\end{equation*}
where $\mathcal{C}(a_j)$ is a counterclockwise contour in the complex plane which encircles $a_j$ but none of the other $a_s$ for $0\leq s\leq i$. Here and throughout the paper, when referring to a contour integral, we use the notation 
 $\mathcal{C}(z_1,\dots,z_m)$ to indicate that the integral runs over a counterclockwise contour in the complex plane which encircles all
the points ${z_1,\dots,z_m}$ and \emph{does not encircle any pole of the integrand which is not this list}. The specified $z_s$'s will in general be 
themselves poles of the integrand but it may happen that some of them are not,  in which case they do not influence the value of the integral.
Written this way, we have
\begin{equation}
\begin{split}
U_{i,j}\equiv \sum_{k=0}^{n}L^{-1}_{i,k}A_{k,j}&=\sum_{k=0}^{i}L^{-1}_{i,k}{a_k+j\choose j}\\
& = \prod\limits_{s=0}^{i-1}(a_i-a_s) 
\oint_{\mathcal{C}(a_0, a_1,\dots, a_i)} \frac{\displaystyle{ \frac{1}{j!} \prod_{s=0}^{j-1}(t+j-s)}}{\displaystyle{\prod\limits_{s=0}^{i}(t-a_s)}}\frac{dt}{2{\rm i} \pi}\ ,\\
\end{split}
\label{eq:LA}
\end{equation} 
where the summation over $k$ is automatically achieved by the choice of contour which encircles all the poles of the denominator at $t=a_0, \dots, a_i$. Here we simply used the trivial equality
\begin{equation*}
{a \choose m}= \frac{1}{m!} \prod_{s=0}^{m-1}(a-s)
\end{equation*}
for any integers $a\geq 0 $ and $m\geq 0$ to transform the binomial coefficient into a polynomial in $t$.
Since the contour in \eqref{eq:LA} encircles all the poles of the integrand for finite 
$t$, the value of the integral may be obtained as minus the residue of its integrand at infinity. 
Using
\begin{equation*}
 \frac{\displaystyle{ \frac{1}{j!} \prod_{s=0}^{j-1}(t+j-s)}}{\displaystyle{\prod\limits_{s=0}^{i}(t-a_s)}}\underset{t\to \infty}{\sim} \frac{1}{j!}\ 
 t^{j-i-1}=
\left\{\begin{matrix}O\left(\frac{1}{t^2}\right) & \hbox{for}\ i> j\ \ \\
&\\
\displaystyle{\frac{1}{i!}\times \frac{1}{t}} & \hbox{for}\ i=j\ ,
\end{matrix}
\right.
\end{equation*}
we immediately deduce that $U_{i,j}=0$ for $i>j$ since there is no pole at infinity in this case, hence $U$ is upper triangular as wanted.
\end{proof}
Moreover, we have
\begin{equation}\label{eq:Uii}
U_{i,i}=\frac{1}{i!}\, \prod\limits_{s=0}^{i-1}(a_i-a_s) 
\end{equation}
for $i=0,\dots,n$ since the residue at infinity is $-1/i!$. 

From this latest result, we deduce the following expression for the partition function:
\begin{thm}
The partition function reads
\begin{equation}\label{paf}
Z_n=\prod_{i=0}^{n} \frac{\displaystyle{\prod\limits_{s=0}^{i-1}(a_i-a_s)}}{\displaystyle{\prod\limits_{s=0}^{i-1}(i-s)}}=
\frac{\Delta(a_0,a_1,\dots,a_n)}{\Delta(0,1,\dots,n)}\ ,
\end{equation}
where $\Delta(a_0,a_1,\dots,a_n)$ denotes the Vandermonde determinant:
\begin{equation*}
\Delta(a_0,a_1,\dots,a_n):=\det\left((a_i^{j})_{0\leq i,j \leq n}\right)=\prod_{0\leq i<j \leq n} (a_j-a_i)\ .
\end{equation*}
\end{thm}

\begin{example}
In the particular case $a_i=p\, i$ for some integer $p\geq 1$, this Theorem yields a partition function
\begin{equation*}
Z_n=p^{\frac{n(n+1)}{2}}
\end{equation*}
in agreement with the result of \cite{DFLAP} for $p=2$. Note also that the matrix $L^{-1}$
then has elements $L^{-1}_{i,j}=(-1)^{i+j}{i \choose j}$ independently of $p$. 
\end{example}

To conclude this section, we note that, by consistency, the same expression for the partition function should be obtained upon using any
of the three possible path formulations of Section \ref{tiling}. 
From the LGV lemma, this allows us to express $Z_n$ as the determinant of the matrix $\tilde{A}$ 
of size $(n+1)\times(n+1)$ whose
elements $\tilde{A}_{i,j}$ enumerate paths made of northeast- and east-oriented elementary steps joining $\tilde{O}_i$ to $\tilde{E}_j$, or equivalently 
as the determinant of the matrix $\hat{A}$ of size $m\times m$ whose
elements $\hat{A}_{i,j}$ enumerate paths made of northeast- and north-oriented elementary steps joining $\hat{O}_i$ to $\hat{E}_j$.
The simple combinatorial formulas for $\tilde{A}_{i,j}$ and $\hat{A}_{i,j}$ lead to the identities:
\begin{equation*}
\det{a_i+j\choose j}_{0\leq i,j\leq n}=\det{\tilde{a}_i\choose j}_{0\leq i,j\leq n}
=\det{n+1\choose b_i-j+1}_{1\leq i,j\leq m}= \frac{\Delta(a_0,a_1,\dots,a_n)}{\Delta(0,1,\dots,n)}
\end{equation*}
with $\tilde{a}_i$ as in \eqref{eq:tildeaidef}, $b_i$ as in \eqref{eq:bidef} and $m=a_n-n$ as before.

\section{Tangent method and one-point function: the first piece of the puzzle}
\label{puzzleone}

The aim of this paper is to further study the arctic curve phenomenon, roughly summarized as follows. For large NILP
configurations, two distinct phases can be distinguished: a frozen phase in which paths follow lattice-like regular patterns,
and a liquid entropic phase where paths display more erratic behaviors. It turns out that for special setups, large 
NILP configurations develop a sharp separation between these two phases, along a curve coined ``arctic" for obvious reasons (see Figure 
\ref{fig:tangentmethod} for an illustration).

\subsection{Tangent method and LU decomposition}
\label{tansec}

Let us first describe here the general setting of the \emph{tangent method}, as devised by Colomo and Sportiello \cite{COSPO} 
for the derivation of arctic curves in path models. As opposed to the standard approach consisting in computing bulk expectation values, this method only
requires the knowledge of a much simpler {\it boundary} one-point function. The method goes as follows: we consider NILP configurations with fixed starting and
ending points say $\bv=\{ v_i\}_{i=0,1,\dots,n}$ and $\bw=\{w_j\}_{j=0,1,\dots,n}$ with steps along the oriented edges of some given underlying lattice.
The partition function $Z_{\bv,\bw}$ is given by a LGV determinant: $Z_{\bv,\bw}=\det(A)$,
where the matrix element $A_{i,j}=Z_{v_i,w_j}$ enumerates the possible configurations for a single path joining $v_i$ to $w_j$.
At finite $n$, the NILP configurations for this problem occupy a maximal domain $D$ whose size grows with $n$. 
We may now consider an asymptotic version of the problem with $n$ large, with a suitable rescaling of the underlying lattice so that $D$ 
tends to a scaled domain $\mathcal{D}$ remaining finite when $n\to \infty$.

\begin{figure}
\begin{center}
\includegraphics[width=8cm]{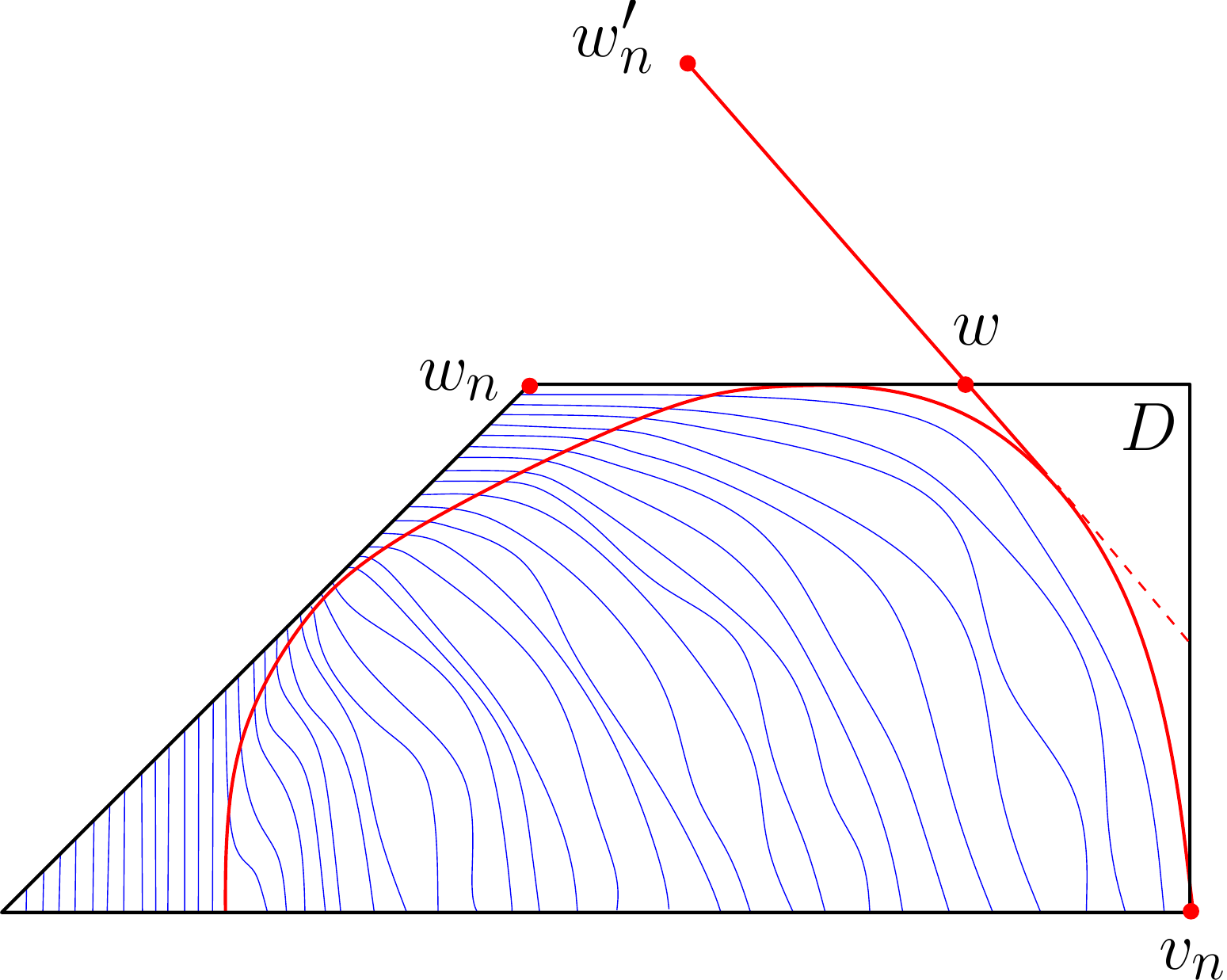}
\end{center}
\caption{A schematic picture of the tangent method: moving the endpoint of the outermost path from $w_n$ to $w'_n$ forces the path
to escape from the originally reachable domain $D$ at some point $w$ on the boundary of $D$. The most likely choice for $w$ is such that
the straight line $(ww'_n)$ is tangent to the arctic curve since the most likely route followed by the outermost path, starting from $v_n$,  consists in sticking
to the arctic curve and escaping from this curve tangentially towards $w'_n$. }
\label{fig:tangentmethod}
\end{figure}

The tangent method relies on the assumption that
\emph{outermost paths} say from $v_n$ to $w_n$ will follow asymptotically the boundary between the frozen and liquid phases of the system,
which sharpens into the arctic curve as $n$ becomes large. To investigate this curve, we simply have to move the endpoint $w_n$ 
to another point $w_n'$ away from $D$ so that paths from $v_n$ to $w_n'$ must escape the domain $D$ (see Figure \ref{fig:tangentmethod}). Let $w$ be the last vertex
of $D$ ($w\in \partial D$) visited by such a path. It is then argued that asymptotically, as it lies away from the influence of the other paths, the escaping
path is most likely to follow a straight line from $w$ to $w_n'$. This line extends within $D$ until the arctic curve is met, and is argued to be 
\emph{tangent} to the latter if we picked for $w$ the most likely escape point from $D$. 
By moving around the new endpoint $w_n'$, we may thus determine lines of most likely escape, which form a parametric family of tangents
to the arctic curve. The latter is then recovered as the envelope of this family of lines.
The modified partition function, normalized by the original one, reads simply $Z_{\bv,\{w_0,\dots,w_{n-1},w_n'\}}/Z_{\bv,\bw}$. By
an asymptotic analysis, we may determine the most likely exit point $w$ from $D$ of the outermost path, which together with 
$w_n'$ defines the tangent line.
This is done is all generality by performing the decomposition
\begin{equation}\label{decompHY}
\frac{Z_{\bv,\{w_0,\dots,w_{n-1},w_n'\}}}{Z_{\bv,\bw}}=\sum_{w\in \partial D} H_{\bv,\bw}^{(w)}  Y_{w,w_n'}\ ,
\end{equation}
where $H_{\bv,\bw}^{(w)}=Z_{\bv,\{w_0,\dots,w_{n-1},w\}}/Z_{\bv,\bw}$ is the so-called \emph{boundary one-point function} in which the outermost 
path ends at $w$ on the boundary of $D$. The last term $Y_{w,w_n'}$ simply enumerates path configurations outside $D$ from $w$ to $w'_n$. 

In practice, the boundary one-point function $H_{\bv,\bw}^{(w)}$ can be computed explicitly by the LU decomposition method \cite{DFLAP}:
first we use for the new partition function the LGV determinant expression $Z_{\bv,\{w_0,\dots,w_{n-1},w\}}=\det(A')$, where the matrix $A'$ differs from $A$
\emph{only in its last column}, which now consists of the partition functions $Z_{v_i,w}$, $i=0,1,\dots,n$.
Assume we found a lower uni-triangular matrix $L$ such that $L^{-1}\, A=U$ is upper triangular. Then, since $A$ and $A'$ differ only in their last column,
$L^{-1}\, A'=U'$ is again upper triangular and differs from $U$ in its last column only. 
We immediately deduce that 
\begin{equation}\label{oneptfromU} H_{\bv,\bw}^{(w)}=\frac{U'_{n,n}}{U_{n,n}} \ .
\end{equation}
As for $Y_{w,w_n'}$, it is in general obtained straightforwardly as it involves configurations of a single path from $w$ to $w_n'$ lying outside $D$,
hence away from the domain of influence of the other paths. 
The most likely exit point $w$ for fixed endpoint $w_n'$ can then be found by an asymptotic analysis of the 
explicit decomposition \eqref{decompHY}, which leads to a parametric family of tangents to the arctic curve.

\subsection{One-point function}
\label{Hnlsec}

Let us now apply the tangent method to our specific problem.
As clear from Figure \ref{fig:path1}, the domain $D$ in which the paths are confined is here a rectangle of vertical size $n$ and horizontal size $a_n$. 
As described above, we now modify the partition function for NILP by moving the topmost endpoint 
$E_n=(0,n)$ along the vertical line to some other position say $E_n'=(0,n+r)$ with a varying $r\in \Z_+$. This choice is somewhat arbitrary but it is easy to check that the final result for the arctic curve would be the same for any other prescription of endpoint that would induce an exit point on 
the segment $(0,n)$--$(a_n,n)$ (for instance by taking $E_n''=(r,n+r)$ instead).

\begin{figure}
\begin{center}
\includegraphics[width=8cm]{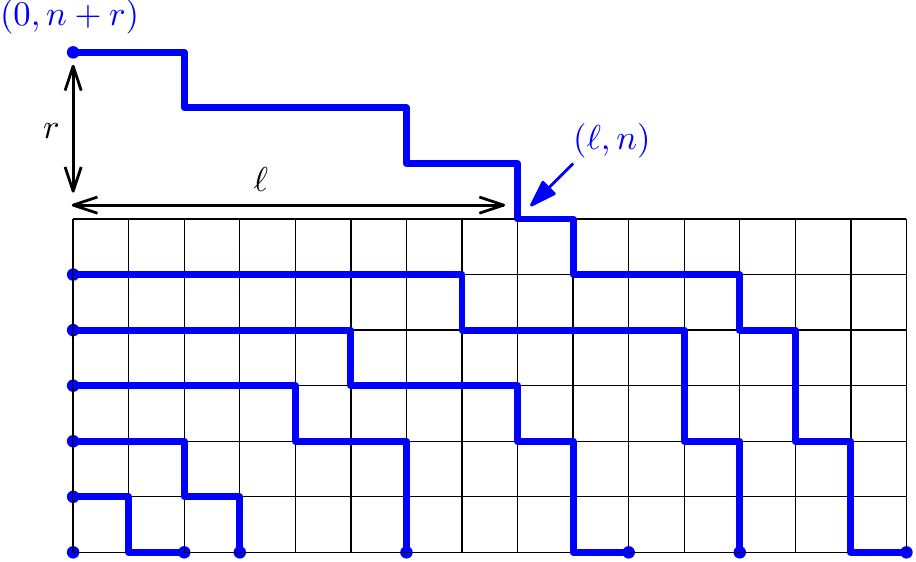}
\end{center}
\caption{The tangent method applied to the NILP under study: the endpoint of the outermost path is moved from $E_n=(0,n)$ to $E'_n=(0,n+r)$ 
with $r\in \Z_+$, forcing the path to escape from the domain $D$ (here the displayed grid) by a north-oriented step at some position $(\ell,n)$ on the boundary of $D$.}
\label{fig:Hnl}
\end{figure}

Let us first compute the one-point function
$H_{n,\ell}$ corresponding to an outermost path from $O_n=(a_n,0)$ exiting at the position $E=(\ell,n)$ from the rectangular
domain $D$ along a north-oriented vertical step $(\ell,n)\to (\ell,n+1)$ pointing out of $D$ (see Figure \ref{fig:Hnl}). 
The LGV matrix $A'$ for such paths reads:
$$A'_{i,j}=\left\{ \begin{matrix} A_{i,j} & {\rm if} \ j<n\ \ \\ & \\  \displaystyle{{a_i+n-l\choose n}} & {\rm if} \ j=n\ .
\end{matrix} \right. $$

\begin{thm}
The one-point function $H_{n,\ell}$ reads:
\begin{equation}\label{exactH}
H_{n,\ell}=\oint_{{\mathcal C}(S_\ell)}  \frac{dt}{2{\rm i} \pi} \prod_{s=0}^n \frac{1}{(t- a_s)} \prod_{s=1}^n (t-\ell+s)\ ,
\end{equation}
where $S_\ell=\{a_s\, \vert a_s\geq \ell\}$.
\end{thm}
\begin{proof}
We use the LU decomposition method with the matrix $L^{-1}$ displayed in \eqref{eq:Lmat}
to compute:
\begin{eqnarray*}
U'_{n,n}&=&\sum_{k=0}^n (L^{-1})_{n,k} A'_{k,n} =\sum_{k=0}^n 
\frac{\displaystyle{\prod\limits_{s=0}^{n-1}(a_n-a_s)}}{\displaystyle{\prod\limits_{s=0\atop s\neq k}^{n}(a_k-a_s)}} {a_k+n-\ell\choose n}  \\
&=&\displaystyle{\prod\limits_{s=0}^{n-1}(a_n-a_s)}  \oint_{{\mathcal C}(S_\ell)} \frac{dt}{2{\rm i} \pi} 
\frac{1}{\displaystyle{\prod\limits_{s=0}^{n}(t-a_s)}} 
\frac{1}{n!} \prod_{s=0}^{n-1}(t+n-\ell-s)\ ,
\end{eqnarray*}
where the contour integral picks up the residues at \emph{all the poles for which the binomial coefficient is 
well-defined and non-zero}, namely at all the points $a_s$ such that $a_s\geq \ell$. The Theorem follows from the identity \eqref{oneptfromU},
by normalizing by $U_{n,n}=\frac{1}{n!}\prod\limits_{s=0}^{n-1}(a_n-a_s)$,
as given by \eqref{eq:Uii}, and changing $s$ into $n-s$ in the last product.

\end{proof}
\begin{remark}
\label{remcontour}
Note that the contour ${\mathcal C}(S_\ell)$ in \eqref{exactH} may be extended into ${\mathcal C}(S_{\ell-n})$ i.e.\ encircle also those $a_s$ between $\ell-n$ and $\ell-1$
since $\prod_{s=1}^n (t-\ell+s)$ vanishes for all integers $t$ in this range.
\end{remark}

Finally, the single path partition function from the exit point $(\ell,n+1)$ to the remote endpoint $(0,n+r)$ is simply
\begin{equation}\label{exactY}Y_{\ell,r}= {\ell+r-1\choose \ell}\ . \end{equation}

\subsection{Asymptotic analysis and arctic curve I}
\label{asymptoone}

We now study the large $n$ asymptotics of the identity \eqref{decompHY} for our model. To this end, let us introduce rescaled
variables
\begin{equation}
 \ell=n\, \xi,\quad r=n\, z, \quad a_i= n \, \al\left(\frac{i}{n}\right)\ ,
 \label{eq:rescaling}
\end{equation}
where $u\mapsto \alpha(u)$ is a fixed piecewise differentiable increasing function from $[0,1]\to \R_+$ encoding the fixed limiting
endpoint distribution. Note that moreover $\al'(u)\geq 1$ whenever the derivative of $\al$ is well-defined due to the 
condition $a_{i+1}-a_i\geq 1$.
\medskip
The main result of this section may be summarized into the following theorem.
\begin{thm}\label{arcticonethm}
The portion of arctic curve obtained with the tangent method for the path setup in which the target endpoint is moved away from $D$ in the northwest corner
and the escape point  is on the top boundary of $D$
has the following parametric representation:
\begin{equation}\label{arcticone}
\left\{ \begin{matrix}
X=X(t):= & t- \frac{\displaystyle{x(t)(1-x(t))}}{\displaystyle{x'(t)}} \hfill\\
& \\
Y=Y(t):= &\frac{\displaystyle{(1-x(t))^2}}{\displaystyle{x'(t)}}\hfill
\end{matrix}\right. \qquad (t\in [\al(1),+\infty))\ ,
\end{equation}
where the quantity $x(t)$ is defined as:
\begin{equation}\label{defx} x(t):= e^{\textstyle{-\int_0^1 \frac{du}{t-\al(u)}}}\ . \end{equation}
\end{thm}
\noindent Here $X$ and $Y$ denote rescaled coordinates in the plane, as obtained by {\emph{after rescaling all coordinates by $n$} so that $D$ becomes a rectangle $\mathcal{D}$ of vertical size $1$ and horizontal size $\alpha(1)$. 
\begin{proof}
The exact formulas \eqref{exactH}-\eqref{exactY} lead to the following leading asymptotic behaviors:
\begin{eqnarray} && \qquad H_{n,n\xi} \sim \oint  \frac{dt}{2{\rm i} \pi} e^{nS_0(t,\xi)}\ ,\qquad Y_{n\xi,n z} \sim e^{nS_1(\xi,z)}\ ,\nonumber  \\
&&\ \ S_0(t,\xi) =\int_0^1 du\, {\rm Log}\left(\frac{t+u-\xi}{t-\al(u)}\right)\label{Szero}\\
&& \quad \quad \quad  =-1+(t+1-\xi){\rm Log}(t+1-\xi)-(t-\xi){\rm Log}(t-\xi)-\int_0^1 du\, {\rm Log}(t-\al(u))\ ,\nonumber  \\
&&\ \  S_1(\xi,z)= (\xi+z){\rm Log}(\xi+z) -\xi{\rm Log}(\xi)-z{\rm Log}(z)\ .\nonumber  
\end{eqnarray}
Note that we performed a harmless rescaling of the integration variable $t\to n t$. In this new variable, the integration contour (originally
${\mathcal C}(S_\ell)$), must encircle the segment $[\xi,\alpha(1)]$. On the left side of this segment, we note, using remark \ref{remcontour}, 
that the contour may cross the real axis anywhere between $\xi-1$ and $\xi$. On the right side, it may cross the real axis at any position $t\in [\al(1),+\infty)$.
At large $n$, the contour integral is evaluated by a simple saddle-point estimate, i.e.\  picking $t$ such that  $\partial_t S_0=0$. Note that
it is important that the saddle-point solution is compatible with the contour constraint. As it will appear, the corresponding value of $t$ is real
and must lie in $[\al(1),+\infty)$.  

The most likely rescaled exit position $\xi$ must maximize the total action\footnote{Indeed, at the saddle-point 
$t=t^*(\xi)$, we have $\frac{d}{d\xi} S(t^*(\xi),\xi,z)=\partial_\xi  S(t^*(\xi),\xi,z)$.} $S(t,\xi,z)=S_0(t,\xi)+S_1(\xi,z)$. 
Writing $\partial_t S_0=\partial_\xi S=0$, we find:
\begin{equation*}
\frac{t+1-\xi}{t-\xi}\ e^{\textstyle{-\int_0^1 \frac{du}{t-\al(u)} }}=1\quad {\rm and}\quad 
\frac{(\xi+z)(t-\xi)}{\xi(t+1-\xi)} =1\ .
\end{equation*}
In terms of the quantity $x(t)$ of \eqref{defx}, this leads to the solution:
\begin{equation*}
\xi=\xi(t):= t- \frac{x(t)}{1-x(t)} \quad {\rm and} \quad z=z(t):=t\,\frac{1-x(t)}{x(t)}-1\ .
\end{equation*}
Clearly, we want $\xi(t)$ and $z(t)$ real, which implies $t$ real. Moreover, we have $(\xi(t)-t)(t-(\xi(t)-1))=-x(t)/(1-x(t))^2 <0$ as $x(t)>0$,
which means that $t$ cannot lie in the interval $[\xi-1,\xi]$. This leaves us with the range $t\in [\al(1),+\infty)$: the result above
is only valid if $t$ lies in this range. Letting $t$ vary from $\alpha(1)$ to $+\infty$ corresponds in turn to letting $x(t)$
increase from $0$ to $1$.

The (tangent) line passing through the rescaled escape point $(\xi(t),1)$ and the rescaled moved endpoint $(0,1+z(t))$ 
is defined by the equation $\xi(t)\, Y+z(t)\,X=\xi(t)(1+z(t))$,
or equivalently
\begin{equation}
x(t)\, Y+(1-x(t))\,(X-t)=0\ .
\label{eq:tgeq}
\end{equation} 
In particular, this allows us to interpret the parameter $t$ as the \emph{intercept of the tangent line with the $X$-axis}. 
The range $t\in [\al(1),+\infty)$
corresponds to negative slopes $-(1-x(t))/x(t)$.
The envelope of this parametric family of lines is obtained by solving the system
\begin{equation*}
\begin{split}
& x(t)\, Y+(1-x(t))\,(X-t)=0\\ 
&x'(t)\, Y-x'(t)\,(X-t)-1+x(t)=0\\ 
\end{split} 
\end{equation*}
and leads immediately to \eqref{arcticone}.
\end{proof}

Let us stress again that, due to the setup that we have used for applying the tangent method, namely that we decided to move the topmost endpoint
$E_n=(0,n)$ to $E'_n=(0,n+r)$, the Theorem \ref{arcticonethm} above 
provides us only with a {\it portion}  of the arctic curve. Other portions will be studied below. Let us examine the limiting points of 
the current portion: in the limit $t\to\infty$ ($x(t)\to 1$), we have the expansion
$$x(t)= 1-\frac{1}{t}+\frac{1}{t^2}\left(\frac{1}{2}-\int_0^1 \al(u) du\right)+O\left(\frac{1}{t^3}\right)\ , $$
hence the limiting point on the arctic curve has coordinates $(X_1,Y_1)$ with
\begin{equation}
 X_1=\frac{1}{2}+\int_0^1 \al(u)du,\qquad Y_1=1
 \label{eq:X1}
 \end{equation}
and corresponds to a horizontal tangent. Note that, from the conditions $\al(0)=0$ and $\al'(u)\geq 1$ for all $u$, we deduce 
the bounds $1\leq X_1 \leq \alpha(1)$.
At the other end when $t\to \al(1)$ ($x(t)\to 0$), writing $t=\al(1)+\theta$ for small $\theta$ leads to the estimate:
\begin{eqnarray*}{\rm Log}(x(t))&=&-\int_0^1 \frac{du}{\theta+(1-u)\al'(1)} 
-\int_0^1 du\left\{\frac{1}{\theta+\al(1)-\al(u)}-\frac{1}{\theta+(1-u)\al'(1)} \right\}\\
&=& \frac{1}{\al'(1)}{\rm Log}\left(\frac{\theta}{\al'(1)}\right)-\int_0^1 du\left\{\frac{1}{\al(1)-\al(u)}-\frac{1}{(1-u)\al'(1)} \right\}+O(\theta)\ ,
\end{eqnarray*}
where the subtraction term was devised so that the integral in the second line is finite.
We deduce from \eqref{arcticone} that, since $x(t)/x'(t)\sim \al'(1)\, \theta$,  $X\to \al(1)$ whereas 
$$Y=\frac{(1-x(t))^2}{x'(t)}\simeq \frac{\theta^{1-\frac{1}{\al'(1)}} }{\al'(1)^{1+\frac{1}{\al'(1)}}} 
e^{\textstyle{\int_0^1 du\left\{\frac{1}{\al(1)-\al(u)}-\frac{1}{(1-u)\al'(1)} \right\}}}\ .$$
We see that if $\al'(1)>1$ then $Y\to 0$, and the endpoint of the arctic curve has coordinates $(X_0,Y_0)=(\al(1),0)$ with a vertical tangent.
On the other hand, if $\al'(1)=1$, then $Y$ has a finite limit, and the endpoint is:
\begin{equation}
\label{eq:X0Y0}
X_0=\al(1),\quad Y_0=e^{\textstyle{\int_0^1 du\left\{\frac{1}{\al(1)-\al(u)}-\frac{1}{(1-u)} \right\}}}
\end{equation}
with a vertical tangent. The case where $\al'(u)=1$ on a finite interval $[1-\gamma,1]$ will be treated in Section \ref{freezesec} below. 

The above discussion assumed implicitly 
that $\al'(1)$ is finite. For $\al'(1)=+\infty$, we must consider the two integrals $I_1=\int_0^1 \frac{du}{\al(1)-\al(u)}$
and $I_2=\int_0^1 \frac{du}{(\al(1)-\al(u))^2}$.
Assuming the behavior $\al(1)-\al(u)\sim C (1-u)^a$ for $0<a<1$, we see that both $I_1$ and $I_2$ are finite for $a<\frac{1}{2}$, while $I_1$ is finite positive and $I_2$ diverges for $a\geq \frac{1}{2}$. 
When both $I_1$ and $I_2$ are finite,
we have $\lim_{t\to \al(1)} x(t)=e^{-I_1}<1$ and $\lim_{t\to \al(1)} \frac{x'(t)}{x(t)}=I_2>0$. This leads to the endpoint
$$X_0=\al(1)-\frac{1-e^{-I_1}}{I_2},\quad Y_0=\frac{(1-e^{-I_1})^2}{I_2\, e^{-I_1}}\ ,$$
with a tangent of negative slope $\lim_{t\to \al(1)} (x(t)-1)/x(t)=1-e^{I_1}$ so that the arctic curve is tangent to the line connecting $(X_0,Y_0)$ to
$(\al(1),0)$.
When $I_1$ is finite and $I_2$ diverges, this leads as before to an endpoint $(X_0,Y_0)=(\al(1),0)$ 
but with now a finite negative slope $1-e^{I_1}$.

\begin{example}
To illustrate our result, we display in Figure \ref{fig:tangentshalfp3} the portion of arctic curve given by \eqref{arcticone} in the particular case $\al(u)=3u$
together with some set of tangents enveloping this curve. In this case $x(t)=\left(\frac{t-3}{t}\right)^{1/3}$ from \eqref{defx}, $(X_0,Y_0)=(3,0)$ and
$(X_1,Y_1)=(2,1)$.

\begin{figure}
\begin{center}
\includegraphics[width=10.5cm]{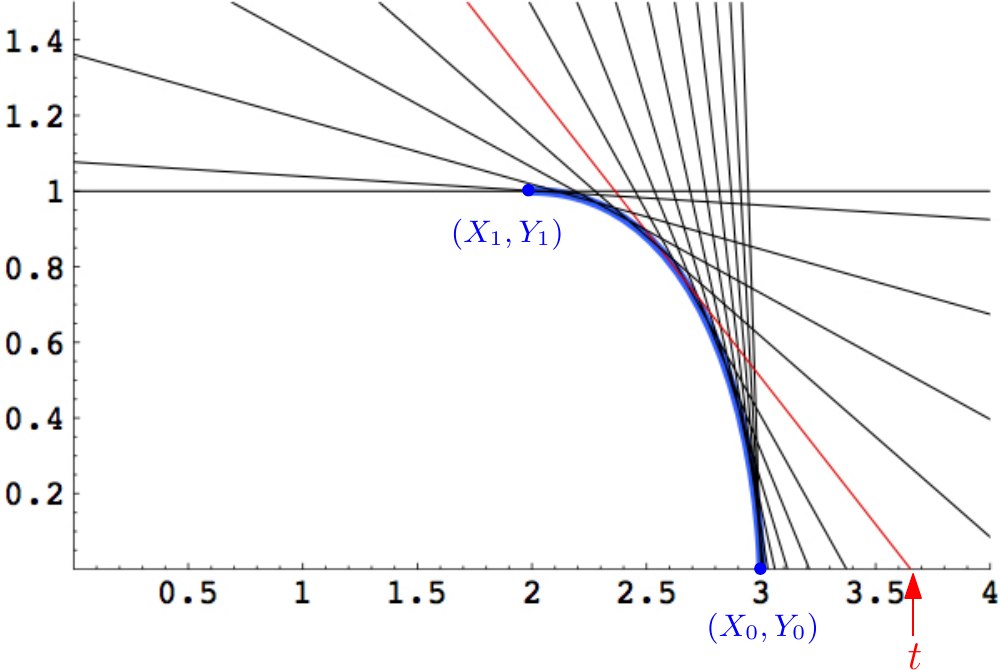}
\end{center}
\caption{The portion of arctic curve given by \eqref{arcticone} in the particular case $\al(u)=3u$ (thick solid line), with extremities $(X_0,Y_0)=(3,0)$ and
$(X_1,Y_1)=(2,1)$. We also displayed a set of tangents enveloping this curve, as given by \eqref{eq:tgeq} for values of $t$ in the range $[\al(1),+\infty)$ (here $\al(1)=3$). The parameter $t$ corresponds to the abscissa of the intersection point of the tangent with the $X$-axis.}
\label{fig:tangentshalfp3}
\end{figure}
\end{example}

\section{The second piece of the puzzle}
\label{puzzle2}
As we just mentioned, Theorem \ref{arcticonethm} 
solves only one part of the puzzle by providing only a portion of the arctic curve, corresponding to an $X$-coordinate larger than
$X_1$, as given by \eqref{eq:X1}. Let us now derive a second portion of the arctic curve, corresponding
to $X$-coordinates smaller than $X_1$. This is done by repeating the tangent method analysis, now applied to the 
second family of NILP, made of northeast- and east-oriented elementary steps.

\subsection{A simple reflection principle}
\label{reflection}

\begin{figure}
\begin{center}
\includegraphics[width=8.5cm]{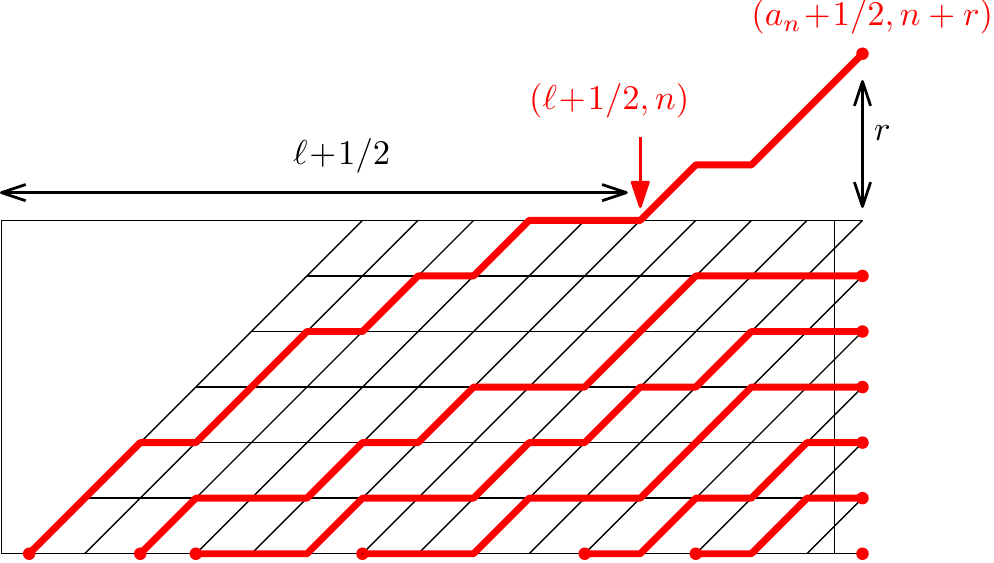}
\end{center}
\caption{The tangent method applied to NILP made of east- and northeast-oriented steps: the endpoint of the outermost path is moved from $\tilde{E}_n=(a_n+1/2,n)$ to $E'_n=(a_n+1/2,n+r)$ 
with $r\in \Z_+$, forcing the path to escape from the domain $D$ (displayed rectangle) by a northeast-oriented step at some position $(\ell+1/2,n)$ on the boundary of $D$.}
\label{fig:Hnltilde}
\end{figure}

Let us consider the equivalent formulation of our problem in terms of the second family of paths. 
These paths, made of northeast- and east-oriented elementary steps, connect
starting points $\tilde{O}_i$ of coordinates $(a_n+1/2-\tilde{a}_i,0)$, with $\tilde{a}_i$ as in \eqref{eq:tildeaidef}, to endpoints
$\tilde{E}_i$ of coordinates $(a_n+1/2,i)$,  for $i=0,\dots,n$. 
We may again apply the tangent method and compute the one-point function
$\tilde{H}_{n,\ell}$ corresponding to an outermost path starting from $\tilde{O}_n=(1/2,0)$ and escaping at the position $\tilde{E}=(\ell+1/2,n)$ from the rectangular domain $D$ along a northeast-oriented diagonal step $(\ell+1/2,n)\to (\ell+3/2,n+1)$ pointing out of $D$ (see Figure \ref{fig:Hnltilde}). 
Note that, since elementary steps are northeast- or east-oriented, the smallest possible $X$-coordinate for the escape point is $(n+1/2)$ hence we have 
now the condition $\ell\geq n$.
The escape path is then eventually extended to a new endpoint, say $\tilde{E}_n'=(a_n+1/2,n+r)$, $r\in \Z_+$, corresponding to moving the original
endpoint $\tilde{E}_n$ by $r$ elementary steps to the north.
The single path partition function from the exit point $(\ell+3/2,n+1)$ to the remote endpoint $(a_n+1/2,n+r)$ is simply
\begin{equation}\label{exacttildeY}\tilde{Y}_{\ell,r}= {a_n-\ell-1\choose r-1}\ . \end{equation}

\begin{figure}
\begin{center}
\includegraphics[width=10cm]{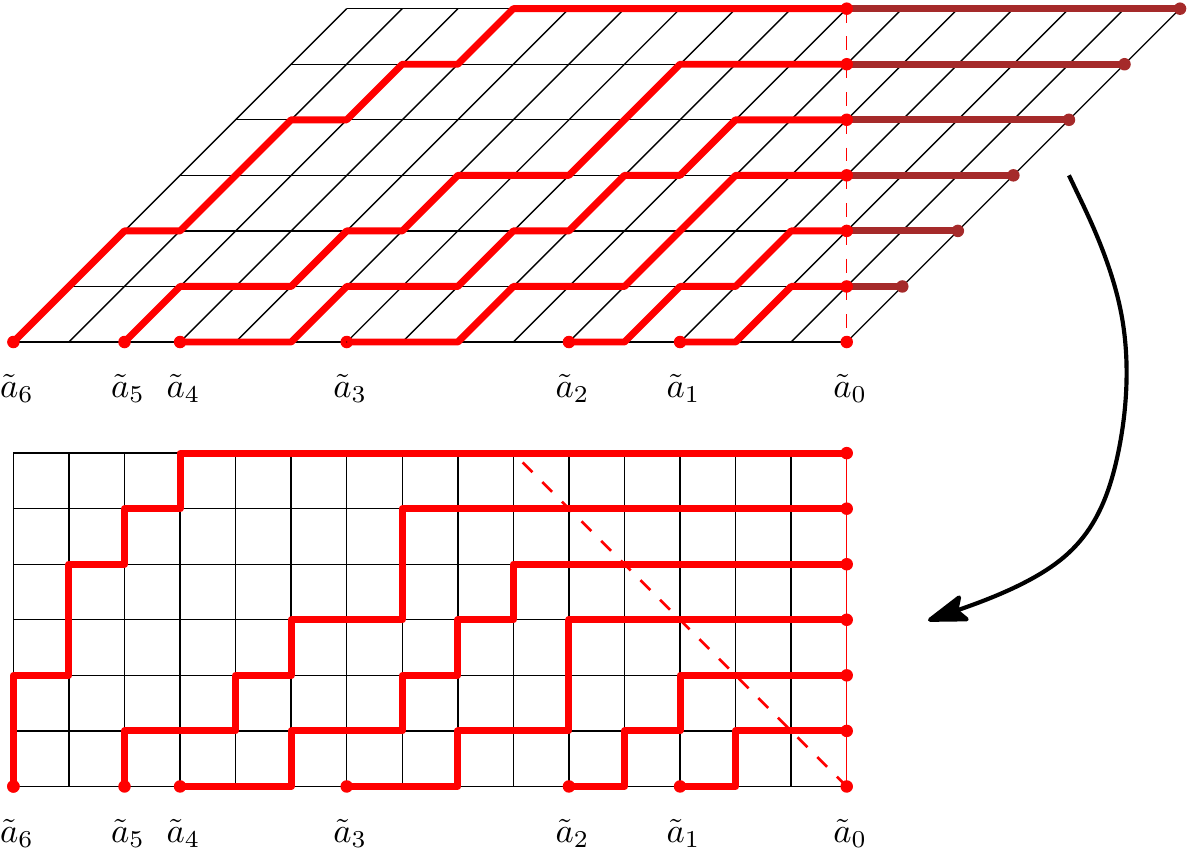}
\end{center}
\caption{For a NILP configuration with paths made of east- and northeast-oriented steps,  moving the endpoints $\tilde{E}_i$ from position 
$(a_n+1/2,i)$ to position $(a_n+1/2+i,i)$ does not modify the enumeration problem since the added portions of path are frozen into horizontal segments. A simple 
shear transforms this extended NILP into a NILP made of north- and east-oriented steps which, upon reflection is of the same type as that of Figure \ref{fig:path1} up to the change of sequence from $(a_i)_{0\leq i\leq n}\to (\tilde{a}_i)_{0\leq i\leq n}$. }
\label{fig:reflection}
\end{figure}

As for the new one-point function, we have the following theorem:

\begin{thm}
The one-point function $\tilde{H}_{n,\ell}$ ($\ell\geq n$) reads:
\begin{equation}\label{exacttildeH}
\tilde{H}_{n,\ell}=-
\oint_{{\mathcal C}(S_0\setminus S_{\ell-n+1})}  \frac{dt}{2{\rm i} \pi} \prod_{s=0}^n \frac{1}{(a_s-t)} \prod_{s=0}^{n-1} (\ell-t-s)\ ,
\end{equation}
where $S_0\setminus S_{\ell-n+1}=\{a_s\, \vert a_s\leq \ell-n\}$.
\end{thm}

\begin{proof}
Let us show how to derive the expression of $\tilde{H}_{n,\ell}$ directly from our previous result for $H_{n,\ell}$ 
via a simple \emph{reflection principle}.
As displayed in Figure \ref{fig:reflection}, the endpoints $\tilde{E}_i$ of coordinates $(a_n+1/2,i)$ for the second family of paths can be moved 
toward east to position $(a_n+1/2+i,i)$ without changing the path enumeration problem. Indeed, the constraint of non-intersection
of the paths forces the path extensions to form straight horizontal segments. The obtained configuration may then be transformed
into a set of north- and east-oriented NILP on a square grid by the simple (shear) mapping $(X,Y)\mapsto (X-Y,Y)$ (see  Figure \ref{fig:reflection}).
Up to a reflection $(X,Y)\to (1/2+a_n-X,Y)$, we immediately recognize the setting of our first set of NILP (made of north- and west-oriented
elementary steps), where the strictly increasing sequence $(a_i)_{0\leq i\leq n}$ is simply replaced by the strictly increasing sequence 
$(\tilde{a}_i)_{0\leq i\leq n}$. This identification holds also in the presence of some escape point for the uppermost path. If
this point has coordinates $(\ell+1/2,n)$ as in Figure \ref{fig:Hnltilde}, its $X$-coordinate is transformed by the two successive mappings above 
(shear and reflection) and takes the
value $\tilde{\ell}=a_n-\ell+n$. We may therefore transpose the expression \eqref{exactH} for $H_{n,\ell}$ and write directly, without new calculation,
 \begin{equation*}
\tilde{H}_{n,\ell}=\oint_{{\mathcal C}(\tilde{S}_{\tilde{\ell}})}  \frac{dt}{2{\rm i} \pi} \prod_{s=0}^n \frac{1}{(t- \tilde{a}_s)} \prod_{s=1}^n (t-\tilde{\ell}+s)
=\oint_{{\mathcal C}(\tilde{S}_{\tilde{\ell}})}  \frac{dt}{2{\rm i} \pi} \prod_{s=0}^n \frac{1}{(t\!-\!a_n\!+\!a_{n-s})} \prod_{s=1}^n (t-a_n+\ell-n+s)\ ,
\end{equation*}
where $\tilde{S}_{\tilde\ell}=\{\tilde{a}_s\, \vert \tilde{a}_s\geq \tilde{\ell}\}$.
Performing the change of variable $t\mapsto a_n-t$ (and changing $s\to n-s$ in both products), we immediately obtain \eqref{exacttildeH}. 
Indeed, after changing variable, the contour explored by the (new) $t$ variable must encircle the $a_n-\tilde{a}_{s}$ such that 
$\tilde{a}_{s}\geq \tilde{\ell}$ hence, using $\tilde{\ell}=n+a_n-\ell$ and $\tilde a_s=a_n-a_{n-s}$ (and changing the dummy variable $s$ into $n-s$), the $a_{s}$ with $a_{s} \leq \ell-n$. This latter set $\{a_s\, \vert a_s\leq \ell-n\}$ is nothing but $S_0\setminus S_{\ell-n+1}$.
\end{proof}
As before, we have the following remark:
\begin{remark}
\label{remcontourtilde}
The contour ${\mathcal C}(S_0\setminus S_{\ell-n+1})$ in \eqref{exacttildeH} may be  extended to ${\mathcal C}(S_0\setminus S_{\ell+1})$ i.e.\ encircle only those $a_s$ between $0$ and $\ell$. Indeed, the integrand in \eqref{exacttildeH} vanishes for all integers $t$ between 
$\ell-n+1$ and $\ell$. 
\end{remark}

\subsection{A combinatorial sum rule}
\label{combHHtile}
Before we discuss the asymptotics of $\tilde{H}_{n,\ell}$ and the associated tangent method result, let us make some comment on the close relation between the one-point functions $\tilde{H}_{n,\ell}$ and $H_{n,\ell}$.
From their expressions \eqref{exactH} an \eqref{exacttildeH}, we deduce the equality, for $\ell \geq n+1$,
\begin{equation*}
\begin{split}
H_{n,\ell}+\tilde H_{n,\ell-1}&=\oint_{{\mathcal C}(S_\ell)}  \frac{dt}{2{\rm i} \pi} \prod_{s=0}^n \frac{1}{(t\!-\!a_s)} \prod_{s=1}^n (t\!-\!\ell\!+\!s)
-\oint_{{\mathcal C}(S_0\setminus S_{\ell})}  \frac{dt}{2{\rm i} \pi} \prod_{s=0}^n \frac{1}{(a_s\!-\!t)} \prod_{s=0}^{n-1} (\ell\!-\!1\!-\!t\!-\!s)\\
&=\oint_{{\mathcal C}(S_0)}  \frac{dt}{2{\rm i} \pi} \prod_{s=0}^n \frac{1}{(t- a_s)} \prod_{s=1}^n (t\!-\!\ell\!+\!s)\\
\end{split}
\end{equation*}
where, using Remark \ref{remcontourtilde}, we extended the contour for $\tilde{H}_{n,\ell-1}$ from ${\mathcal C}(S_0\setminus S_{\ell-n})$ to 
${\mathcal C}(S_0\setminus S_{\ell})$. The final contour
 ${\mathcal C}(S_0)$ encircles all the $a_s$, $s=0,\dots,n$, hence all the (finite) poles of the integrand. The integral may thus be computed as minus the residue at infinity. At large $t$, the integrand behaves as $1/t$, hence the residue is $-1$, leading to the sum rule
\begin{equation}
H_{n,\ell}+\tilde H_{n,\ell-1}=1\ .
\label{eq:sumrule}
\end{equation}
\begin{figure}
\begin{center}
\includegraphics[width=8cm]{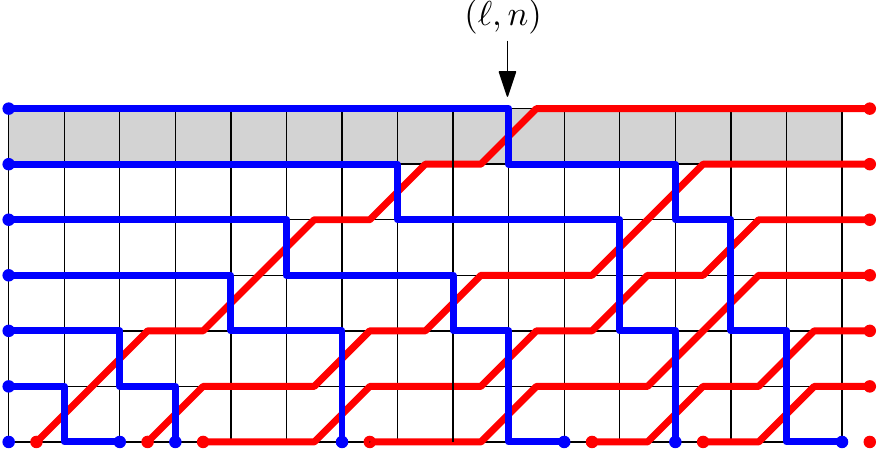}
\end{center}
\caption{A configuration of NILP in $D$. In the language of paths made of north- and west-oriented steps, the configuration 
has a unique vertical step in the uppermost horizontal strip of $D$, which is part of the outermost path and leads to position $(\ell,n)$ on the boundary. In the equivalent description by east- and northeast-oriented step paths, this step is dual to the unique northeast-oriented step in
the uppermost strip, itself part of the outermost path and leading to position $(\ell+1/2,n)$ on the boundary.}
\label{fig:HHtilde}
\end{figure}
This sum rule has a nice combinatorial interpretation, which we explain now. 
In the original setting with north- and west-oriented step paths, the quantity $Z_n\, H_{n,\ell}$ enumerates configurations where the
$n$'th path exits the domain $D$ by a north-step starting at position $(\ell,n)$. Alternatively, $Z_n\, H_{n,\ell}$ may be interpreted
as configurations where the $n$'th path goes from $O_n$ to $E_n$, hence \emph{remains in the domain $D$ but is required to pass via 
the position $(\ell,n)$}. Indeed, once the position $(\ell,n)$ is reached, the path from $(\ell,n)$ to $(0,n)$ is uniquely determined, made of a straight horizontal
segment of length $\ell$. The quantity $Z_n\, H_{n,\ell}-Z_n\, H_{n,\ell+1}$ therefore enumerates NILP in $D$ where the $n$'th path passes 
via  $(\ell,n)$ but not via $(\ell+1,n)$. This path necessarily reaches $(\ell,n)$ by a north step $(\ell,n-1)\to (\ell,n)$, which is moreover the unique 
vertical step in the uppermost horizontal strip of $D$ (i.e.\ the subdomain of $D$ with $Y$-coordinate between $n-1$ and $n$), see Figure \ref{fig:HHtilde}.
Using now the equivalent description by east- and northeast-oriented step paths, the corresponding $n$'th path in this set necessarily 
has a northeast-oriented step from $(\ell-1/2,n-1)$ to $(\ell+1/2,n)$ hence reaches position $(\ell+1/2,n)$ without passing via position 
$(\ell-1/2,n)$. By the same 
argument as above, configurations satisfying this requirement are enumerated by $Z_n\, \tilde{H}_{n,\ell}-Z_n\, \tilde{H}_{n,\ell-1}$. 
Using this bijective correspondence and simplifying by $Z_n$, we deduce the identity
\begin{equation*}
H_{n,\ell}-H_{n,\ell+1} = \tilde{H}_{n,\ell}-\tilde{H}_{n,\ell-1}\quad 
\Leftrightarrow \quad H_{n,\ell+1}+\tilde{H}_{n,\ell}=H_{n,\ell} +\tilde{H}_{n,\ell-1}\ .
\end{equation*}
This equality states that the quantity $H_{n,\ell} +\tilde{H}_{n,\ell-1}$ \emph{does not depend on $\ell$}, and remains valid for $\ell = n$ with the convention that $\tilde{H}_{n,n-1}=0$ since the outermost path in the second path family setting 
cannot pass via the vertex $(n-1,n)$. Note that $H_{n,n}=1$ (since the outermost path in the original path family setting necessarily passes through the
vertex $(n,n)$) so that the actual common value of $H_{n,\ell} +\tilde{H}_{n,\ell-1}$ for all $\ell\geq n$ is $1$. This is precisely 
the sum rule \eqref{eq:sumrule}.

\subsection{Asymptotic analysis and arctic curve II}
\label{asymptotwo}
Applying now the tangent method to the second family of paths, we may complete Theorem \ref{arcticonethm} by the following 
statement:
\begin{thm}\label{arctictwothm}
The portion of arctic curve obtained with the tangent method for the path setup in which the target endpoint is moved away from $D$ in the northeast corner
and the escape point  is on the top boundary of $D$
has the following parametric representation:
\begin{equation}\label{arctictwo}
\left\{ \begin{matrix}
X=X(t):= & t- \frac{\displaystyle{x(t)(1-x(t))}}{\displaystyle{x'(t)}} \hfill\\
& \\
Y=Y(t):= &\frac{\displaystyle{(1-x(t))^2}}{\displaystyle{x'(t)}}\hfill
\end{matrix}\right. \qquad (t\in (-\infty,0])\ ,
\end{equation}
with $x(t)$ as in \ref{arcticonethm}.
\end{thm}
In other words, the arctic curve parametrization of Theorem \ref{arcticonethm} extends to values of $t$ in $(-\infty,0]$, leading to a new 
portion of the arctic curve which we will describe below.

\begin{proof}
Using the same rescaling \eqref{eq:rescaling} as in Section \ref{asymptoone}, we now get from the exact formulas \eqref{exacttildeH}-\eqref{exacttildeY} the asymptotic behaviors, valid for $\xi\geq 1$ (recall that $\ell \geq n$ in $\tilde{H}_{n,\ell}$):
\begin{eqnarray*} && \qquad \tilde{H}_{n,n\xi} \sim - \oint  \frac{dt}{2{\rm i} \pi} e^{n\tilde{S}_0(t,\xi)}\ ,\qquad \tilde{Y}_{n\xi,n z} \sim e^{n\tilde{S}_1(\xi,z)}\ ,\\
\tilde{S}_0(t,\xi) &=&\int_0^1 du\, {\rm Log}\left(\frac{\xi-t-u}{\al(u)-t}\right)\\
&=& -1-(\xi-t-1){\rm Log}(\xi-t-1)+(\xi-t){\rm Log}(\xi-t)-\int_0^1 du\, {\rm Log}(\al(u)-t)\ , \\
 \tilde{S}_1(\xi,z)&=& (\al(1)-\xi){\rm Log}(\al(1)-\xi) -z{\rm Log}(z)-(\al(1)-\xi-z){\rm Log}(\al(1)-\xi-z)\ .
\end{eqnarray*}

Here the contour in the (rescaled) $t$ variable must encircle the segment $[0,\xi-1]$ and, using remark \ref{remcontourtilde}, may
cross the real axis anywhere between $\xi-1$ and $\xi$ on the right side of this segment. On the left side, any position $t\in (-\infty,0]$ is acceptable.
As in Section \ref{asymptoone}, the asymptotic evaluation of the contour integral amounts to picking $t$ such that  $\partial_t \tilde{S}_0=0$ which will
produce a real value of $t$ in the interval $(-\infty,0]$. The most likely rescaled exit position $\xi$ is obtained as before by maximizing the total action $\tilde{S}(t,\xi,z)=\tilde{S}_0(t,\xi)+\tilde{S}_1(\xi,z)$. 
Setting $\partial_t \tilde{S}_0=\partial_\xi \tilde{S}=0$ now leads to:
\begin{equation*}
\frac{\xi-t-1}{\xi-t}\, x(t)=1\quad {\rm and}\quad 
\frac{(\xi-t)(\al(1)-\xi-z)}{(\al(1)-\xi)(\xi-t-1)} =1
\end{equation*}
with $x(t)$ as in \eqref{defx}. We deduce
\begin{equation*}
\xi=\xi(t)= t+ \frac{x(t)}{x(t)-1} \quad {\rm and} \quad z=\tilde{z}(t)=(\al(1)-t)\,\frac{x(t)-1}{x(t)}-1\ .
\end{equation*}
Again $t$ must be real and cannot lie in the segment $[\xi-1,\xi]$ and this leaves us with the range $t\in (-\infty,0]$. Letting $t$ vary from $-\infty$ to $0$ corresponds to letting $x(t)$
increase from $1$ to $+\infty$.

The (tangent) line passing through the rescaled escape point $(\xi(t),1)$ and the rescaled endpoint $(\al(1),1+\tilde{z}(t))$ is defined by the equation $(\xi(t)-\al(1))\, Y+\tilde{z}(t)\,X=\xi(t)(1+\tilde{z}(t))-\al(1)$,
or, after substitution and simplification,
\begin{equation}
x(t)\, Y+(1-x(t))\,(X-t)=0\ .
\label{eq:tgeqtwo}
\end{equation} 
Remarkably, the equation for the tangent lines is \emph{the same} as that \eqref{eq:tgeq} in the setting of Section \ref{asymptoone}. Only the
range of $t$, now in the interval $(-\infty,0]$, is changed and corresponds to positive slopes $(x(t)-1)/x(t)$.
The envelope of this new parametric family of lines has therefore the same parametric form \eqref{arcticone} as for Theorem \ref{arcticonethm}
and this leads immediately to \eqref{arctictwo}, hence Theorem \ref{arctictwothm}.
\end{proof}
Again we may examine the limiting points of 
the new portion of arctic curve: in the limit $t\to-\infty$ ($x(t)\to 1$), we recover the point $(X_1,Y_1)$ of \eqref{eq:X1}
with a horizontal tangent.
At the other end of the curve, when $t\to 0$ ($x(t)\to +\infty$), we have the estimate:
\begin{eqnarray*}{\rm Log}(x(t))&=&-\int_0^1 \frac{du}{t-u\, \al'(0)} 
-\int_0^1 du\left\{\frac{1}{t-\al(u)}-\frac{1}{t-u\, \al'(0)} \right\}\\
&=& -\frac{1}{\al'(0)}{\rm Log}\left(\frac{-t}{\al'(0)}\right)-\int_0^1 du\left\{\frac{1}{u\, \al'(0)}-\frac{1}{\al(u)} \right\}+O(t)
\end{eqnarray*}
with a second integral being finite.
We obtain the estimates
\begin{equation*}
\begin{split}
&x(t)\underset{t\to 0^-}{\sim} K \left(\frac{\al'(0)}{-t}\right)^{1/\al'(0)}\ , \\
&x'(t)\underset{t\to 0^-}{\sim} \frac{K}{\al'(0)^2} \left(\frac{\al'(0)}{-t}\right)^{1+1/\al'(0)}\ , \\
&K=e^{\textstyle{-\int_0^1 du\left\{\frac{1}{u\, \al'(0)}-\frac{1}{\al(u)} \right\}}}\ .\\
\end{split}
\end{equation*}
Note that both $x(t)$ and $x'(t)$ tend to $\infty$ for $t\to 0$ with
\begin{equation*}
\frac{x(t)^2}{x'(t)}\underset{t\to 0^-}{\sim} K\, \al'(0)^2 \left(\frac{\al'(0)}{-t}\right)^{1/\al'(0)-1}\ .
\end{equation*}
For $\al'(0)>1$, this ratio tends to $0$ and the endpoint of the arctic curve has coordinates $(X_\infty,Y_\infty)=(0,0)$ with a slope
$1$ since $(x(t)-1)/x(t)$ tends to $1$.
On the other hand, if $\al'(0)=1$, then $X$ and $Y$ have a finite limit, and the endpoint is:
$$X_\infty=Y_\infty=e^{\textstyle{-\int_0^1 du\left\{\frac{1}{u}-\frac{1}{\al(u)} \right\}}}$$
with again a slope $1$. Since the paths cannot enter the domain $Y>X$, the arctic curve is naturally extended from $(X_\infty,Y_\infty)$
to $(0,0)$ by a segment.
The case where $\al'(u)=1$ on a finite interval $[0,\gamma]$ is special in this respect, and will be discussed in Section \ref{freezesec} below. 

The above discussion assumed implicitly 
that $\al'(0)$ is finite. 
For $\al'(0)=+\infty$, we have to be more precise.  Let us assume the behavior $\al(u)\sim C u^a$ when $u\to 0$
with $0<a<1$. We have to consider the two integrals $J_1=\int_{0}^1 \frac{du}{\al(u)}$
and $J_2=\int_{0}^1 \frac{du}{\al(u)^2}$. We note that for $a<\frac{1}{2}$ both integrals are finite, 
while for $a\geq \frac{1}{2}$, $J_1$ is finite and $J_2$ diverges.
If both integrals are finite, then $\lim_{t\to 0} x(t)=e^{J_1}$ and
$\lim_{t\to 0} \frac{x'(t)}{x(t)}= J_2$, and we find the endpoint for $t\to 0$:
$$ X_\infty=\frac{e^{J_1}-1}{J_2} , \quad Y_{\infty}=\frac{(e^{J_1}-1)^2}{J_2\, e^{J_1}} $$
with a tangent of positive slope $\lim_{t\to 0}(x(t)-1)/x(t)=1-e^{-J_1}<1$ so that the arctic curve is tangent to the line connecting $(X_\infty,Y_\infty)$
to $(0,0)$.
If $J_2$ diverges and $J_1$ is finite, then $(X_\infty,Y_\infty)=(0,0)$, and the tangent at the origin has slope
$1-e^{-J_1}<1$.

As a final remark, we note that when the starting point pattern is symmetric by reflection, i.e.\ whenever $\tilde{a}_i=a_i$, hence
$\al(u)=\al(1)-\al(1-u)$,
the arctic curve is 
symmetric under the involution $(X,Y)\mapsto (\alpha(1)-X+Y,Y)$ as a direct consequence of the reflection principle detailed in Section \ref{reflection} above. This is visible 
in the parametric equation of the curve: indeed, using $\al(u)=\al(1)-\al(1-u)$, we get the identity
$x(\al(1)-t)=1/x(t)$. Plugged into the parametric equation, it yields $X(\al(1)-t)=\al(1)-X(t)+Y(t)$ and $Y(\al(1)-t)=Y(t)$.
The above symmetry of the arctic curve is therefore associated with the involution $t\mapsto \al(1)-t$ for the parameter $t$ .

\section{Examples}
\label{examplesec}

In this section, we present various examples to illustrate the general results of Sections \ref{puzzleone} and  \ref{puzzle2} above.
As a preliminary remark, we note that any continuous piecewise differentiable increasing function $\alpha(u)$ on $[0,1]$ with $\al'(u)\geq 1$ (when it is defined)
may be realized by taking starting points $(a_i,0)$ with
\begin{equation}
a_i=\Big\lfloor n\, \al\left(\frac{i}{n}\right)\Big\rfloor \ .
\label{eq:floor}
\end{equation}  
The condition $\al'(u)\geq 1$ guarantees that this sequence is indeed strictly increasing\footnote{As we shall see later, it is interesting to also address the case where $\al(u)$ presents discontinuities with
positive jumps $\delta_k$. In that case, eq.~\eqref{eq:floor} is only valid for large enough $n\geq \max_k(1/\delta_k)$ to ensure
that the sequence $(a_i)$ is strictly increasing.} and its scaling limit is clearly described
by $\al(u)$.

\subsection{The pure case $\al(u)=p\, u$}

We consider the case where $\al(u)=p\, u$ for some real number $p>1$. 
For instance, the particular case $p\in \N\setminus\{1\}$ is obtained as the large $n$ limit of the points $a_i=p\, i$, $i=0,1,\dots,n$.

Substituting $\al(u)=p\, u$ into \eqref{defx} yields
\begin{equation} 
x(t)=e^{\textstyle{-\int_0^1 \frac{du}{t- p u}}}= \left(1-\frac{p}{t}\right)^{\frac{1}{p}}\ .
\end{equation}
The two portions of the arctic curve correspond respectively to $t\in (-\infty,0]$ and $t\in [p,+\infty)]$, namely to $x(t)\in [0,+\infty)$.
More precisely, we may express the arctic curve of Theorems \ref{arcticonethm} and \ref{arctictwothm} in terms of the parameter $x\equiv x(t)$, by noting that 
$t=p/(1-x^p)$ and $x'(t)= (1-x^p)^2/(p^2\, x^{p-1})$ as:
\begin{equation}\left\{
\begin{split}
&X=\frac{p}{(1-x^{p})}\left( 1-\frac{p\, (1-x)}{(1-x^{p})}\, x^{p}\right)\\
&Y=\frac{p^2(1-x)^2}{(1-x^{p})^2}\, x^{p-1}
 \end{split}\right. (x\in [0,+\infty))\ .
\label{purep}\end{equation}
The special points on the curve, corresponding respectively to $x=+\infty, 1,0$, are the origin $(X_\infty,Y_\infty)=(0,0)$ with a tangent of slope $1$, 
the maximum $(X_1,Y_1)=(\frac{p+1}{2},1)$ with horizontal tangent and the endpoint $(X_0,Y_0)=(p,0)$ with vertical tangent.
When $p$ is an integer, eq.\eqref{purep} may be recast into:
\begin{equation}\left\{
\begin{split}
&X=\frac{p\, (1+2\,x+3\,x^2+\cdots +p\, x^{p-1})}{(1+x+x^2+\cdots +x^{p-1})^2}\\
&Y=\frac{p^2\, x^{p-1}}{(1+x+x^2+\cdots +x^{p-1})^2}
 \end{split}\right. (x\in [0,+\infty))\ .
\end{equation}
For $p=2$, this simplifies drastically, as we may eliminate $x=Y/(2(X-Y))$, and we recover the arctic parabola of \cite{DFLAP}:
$$ (2X-Y)^2-8(X-Y)=0\ .$$
For $p=3$, eliminating $x$ leads to the following quartic arctic curve:
$$  (3 X^2 - 3 X Y + Y^2)^2 -2 (3 X - Y) (9 X^2 - 15 X Y + 7 Y^2)+81 (X - Y)^2 =0\ .$$
\begin{figure}
\begin{center}
\includegraphics[width=14cm]{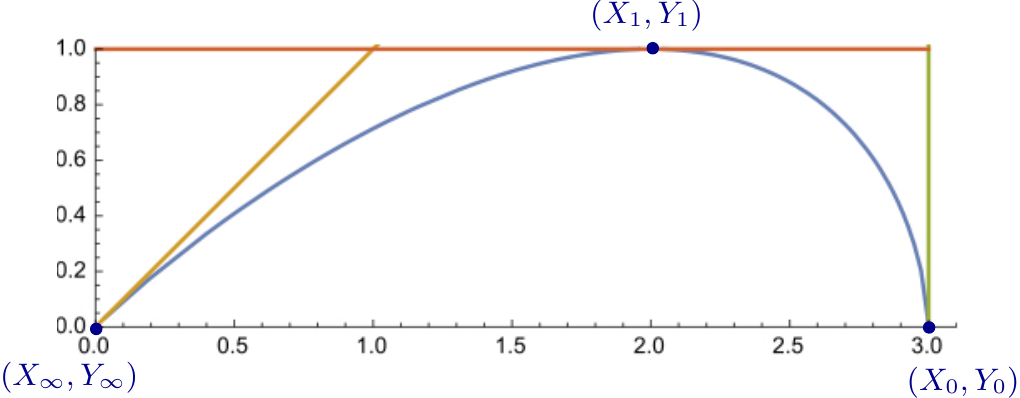}
\end{center}
\caption{The arctic curve in the case $\al(u)=3\, u$. The slope is horizontal at $(X_1,Y_1)=(2,1)$ on the upper boundary of $\mathcal{D}$, vertical  at $(X_0,Y_0)=(3,0)$ on the right boundary of $\mathcal{D}$, and $1$ at $(X_\infty,Y_\infty)=(0,0)$ on the left boundary of $\mathcal{D}$ so that 
the arctic curve is tangent to the indicated line $Y=X$ (above which paths are fully frozen even for finite $n$).}
\label{fig:p3}
\end{figure}
The corresponding curve is displayed in Figure \ref{fig:p3} for illustration.
For higher integer values of $p$, by eliminating $x$, one can show that the arctic curve is an algebraic curve of degree $2p-2$.
The case of rational $p\geq 1$ also leads to an algebraic arctic curve. For instance, for $p=3/2$ we find:
\begin{equation*}
\begin{split}
& 32 (3 X^2 - 3 X Y + Y^2)^2 -16 (54 X^3 - 135 X^2 Y + 99 X Y^2 - 19 Y^3)\\
&\qquad \qquad \qquad \qquad \qquad \qquad+162 (5 X - 8 Y) (X - Y)-243 (X - Y)=0 \ .\\
\end{split}
\end{equation*}

It is interesting to notice that there is a well-defined large $p$ limit of the arctic curve, provided one rescales the $X$ coordinate by a factor $1/p$.
In the new coordinates $(\tilde X,\tilde Y)=(X/p,Y)$, using the finite parameter $e^y=x^p$, i.e.\ setting ${\rm Log}(x)=\frac{y}{p}$ and letting $p\to\infty$,
we find
\begin{equation}\left\{
\begin{split}
&{\tilde X}=\frac{1}{(1-e^y)}\left( 1+\frac{y\, e^y}{(1-e^y)}\right)=\frac{y}{4\sinh^2(y/2)}-\frac{e^{-y/2}}{2\sinh(y/2)}\\
&{\tilde Y}=\frac{y^2e^y}{(1-e^y)^2}=\frac{y^2}{4\sinh^2(y/2)}
 \end{split}\right. (y\in \R)\ .
\end{equation}
Note the following symmetry: under $y\to -y$, we have $({\tilde X},{\tilde Y})\to  (1-{\tilde X},{\tilde Y})$ so that the arctic curve is symmetric
with respect to the vertical line ${\tilde X}=1/2$. The tangents at the endpoints $(0,0)$
and $(1,0)$ are vertical, while that at the maximum $(\frac{1}{2},1)$ is horizontal.

To end this section, it is interesting to revisit the connection between the asymptotic result for the one-point function $H_{n,\ell}$ and its discrete counterpart. 
Let us for instance consider the case $a_i=3i$ ($p=3$). The one-point function $H_{n,\ell}$ may easily be obtained from the LU decomposition
as 
\begin{equation*}
H_{n,\ell}=\frac{\displaystyle{\sum_{k=\lfloor \ell/3\rfloor}^n (-1)^{k+n} {n \choose k}{3\, k+n-\ell\choose n}}}{\displaystyle{\sum_{k=0}^n (-1)^{k+n} {n \choose k}{3\, k+n\choose n}}}\ .
\end{equation*} 
\begin{figure}
\begin{center}
\includegraphics[width=11cm]{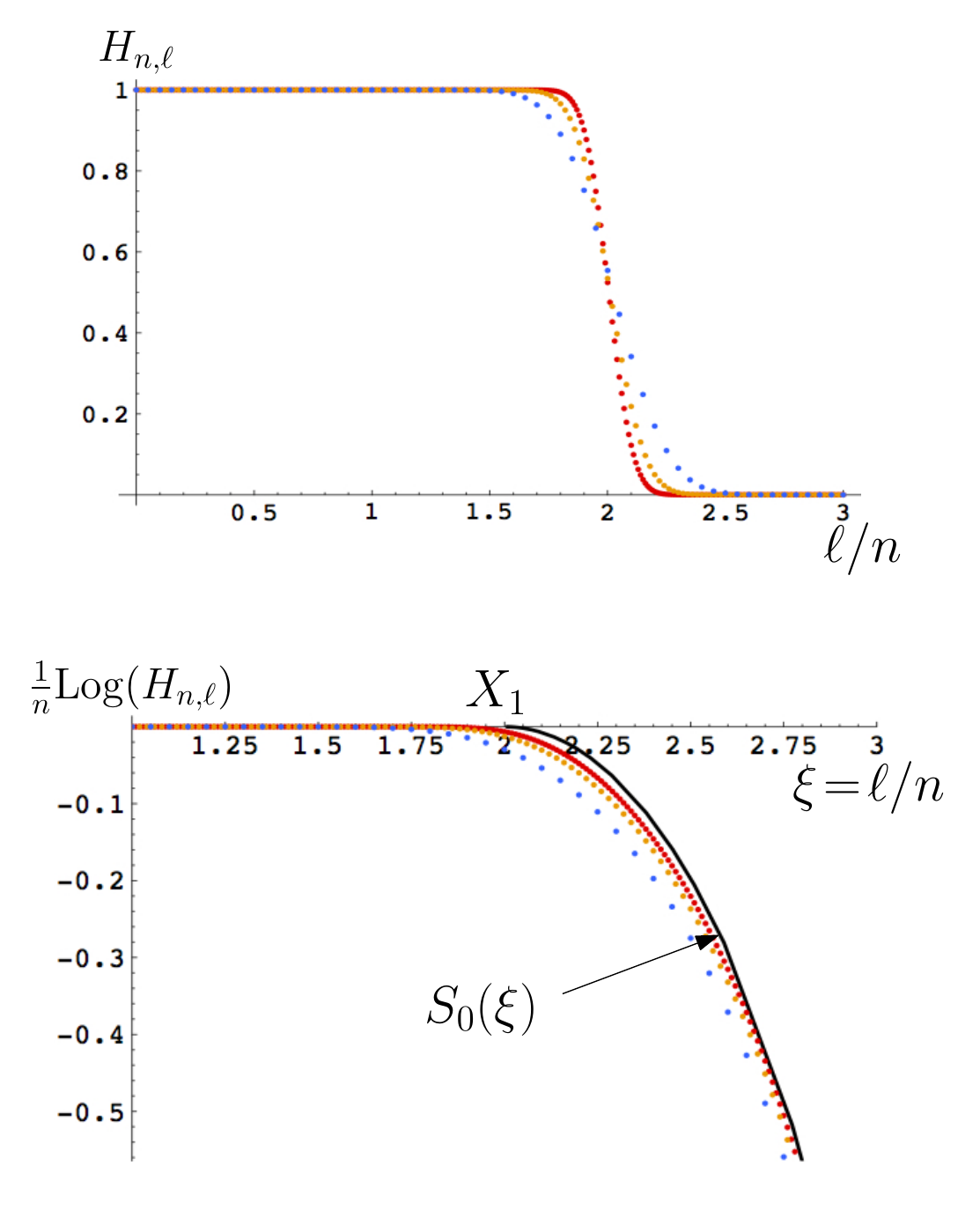}
\end{center}
\caption{Top: the one-point function $H_{n,\ell}$ for the sequence $a_i=3\, i$ and for finite $n=20,50,100$ versus $\ell/n$ presents a sharp transition around $\ell/n=2$ from a limiting value $1$ for small $\ell/n$ to a limiting value $0$ for large $\ell/n$. Bottom: the asymptotic limiting shape of the right part of the transition curve is captured by the quantity 
$\frac{1}{n}{\rm Log}(H_{n,\ell})$ as a function of $\xi=\ell/n$, which, for $\xi\geq X_1=2$, tends at large $n$ towards the scaling function $S_0(\xi)$ of \eqref{eq:S0p3}.}
\label{fig:Hnlscaling}
\end{figure}
 
Figure \ref{fig:Hnlscaling} shows a plot of $H_{n,\ell}$ as a function of $\ell/n$ for increasing values of $n=20,50,100$. We observe a sharp jump
from the value $1$ to the value $0$ taking place at a value of $\ell/n$ tending to $X_1=2$ in this case. The corresponding asymptotics,
describing the large $n$ behavior of $H_{n,\ell}$ for $\ell/n \geq X_1$ is captured by the quantity $\frac{1}{n}{\rm Log}(H_{n,\ell})$ which tends to 
a continuous function $S_0(\xi)$ equal to $S_0(t,\xi)$ of \eqref{Szero} taken at the saddle-point solution $t=t^*(\xi)$ where
$\partial_tS_0(t,\xi)=0$. We find the parametric expression
\begin{equation}
\begin{split}
&\xi=t-\frac{x(t)}{1-x(t)}\\
&S_0(\xi)=\frac{1}{3}(t\!-\!3){\rm Log}(t\!-\!3)-\frac{1}{3}t\, {\rm Log}(t)+\frac{1}{1\!-\!x(t)}{\rm Log}\left(\frac{1}{1\!-\!x(t)}\right)\\&
\qquad \qquad \qquad \qquad \qquad \qquad \qquad \qquad -
\frac{x(t)}{1\!-\!x(t)}{\rm Log}\left(\frac{x(t)}{1\!-\!x(t)}\right)\\
&x(t):=\left(\frac{t-3}{t}\right)^{1/3}\ .\\
\end{split}
\label{eq:S0p3}
\end{equation}
This asymptotic analysis is corroborated by the plot of $\frac{1}{n}{\rm Log}(H_{n,\ell})$ as a function of $\ell/n$ displayed in Figure \ref{fig:Hnlscaling}, for increasing values of $n=20,50,100$, together with the expected limit $S_0(\ell/n)$. The function $S_0(\xi)$ is well defined for $\xi$ between $X_1=2$ ($t\to \infty$) and $3$ ($t\to 3$) and vanishes at $\xi=2$. For $0\leq \xi \leq 2$, the limit of $\frac{1}{n}{\rm Log}(H_{n,\ell})$ vanishes identically, meaning that
$H_{n,\ell}\to 1$ at large $n$ for $\ell\leq 2n$.

\subsection{The case of a piecewise linear $\al(u)$}
\label{sec:piecewise}

Let us consider real numbers  $\gamma_1,\gamma_2,\dots,\gamma_k>0$ such that $\sum_{i=1}^k \gamma_i=1$, and 
real numbers $p_1,p_2,\dots,p_k\geq 1$. We define the function $\al(u)$ to be continuous and piecewise linear with constant derivative
$p_1$ on the interval $[0,\gamma_1]$, $p_{2}$ on $[\gamma_1,\gamma_1+\gamma_2]$, etc.\ , $p_k$ on $[\gamma_1+\cdots+\gamma_{k-1},1]$.
Define variables $\varphi_i:=\sum_{j=1}^{i}\gamma_j$ and $\theta_i:=\sum_{j=1}^i p_j\gamma_j$ for $i=0,1,\dots,k$ with $\varphi_0=0$, $\varphi_k=1$,
and $0=\theta_0<\theta_1<\cdots <\theta_{k-1}<\theta_k=\al(1)$. 
We have for $i=1,2,\dots,k$:
$$ \al(u)= \theta_{i-1} +p_i (u-\varphi_{i-1}) \qquad (u\in [\varphi_{i-1},\varphi_{i}]) \ .$$
The corresponding value of $x(t)$ from \eqref{defx} reads:
\begin{equation}
 x(t)=e^{\textstyle{-\sum\limits_{i=1}^k \int_{\varphi_{i-1}}^{\varphi_i} \frac{du}{t-\theta_{i-1} -p_i (u-\varphi_{i-1})}}}=\prod_{i=1}^k \left(\frac{t-\theta_{i}}{t-\theta_{i-1}}\right)^{\frac{1}{p_i}}=\prod_{i=1}^k \left( 1-\frac{\theta_i}{t}\right)^{\frac{1}{p_{i}}-\frac{1}{p_{i+1}}}
 \label{piecesol}
 \end{equation}
with the convention that $p_{k+1}=+\infty$. 

The maximum with horizontal tangent has coordinates:
$$ X_1=\frac{1}{2}+\sum_{i=1}^k\frac{\theta_i^2-\theta_{i-1}^2}{2\, p_i}\ ,\quad Y_1=1\ .$$
The other special points on the arctic curve depend crucially on the values of $p_1$ and $p_k$. We have $(X_\infty,Y_\infty)=(0,0)$ unless $p_1=1$,
and $(X_0,Y_0)=(\al(1),0)=(\theta_k,0)$ unless $p_k=1$. The situation where either $p_1=1$ or $p_k=1$ is more subtle and will be discussed in Section \ref{freezesec} below.
\begin{figure}
\begin{center}
\includegraphics[width=14cm]{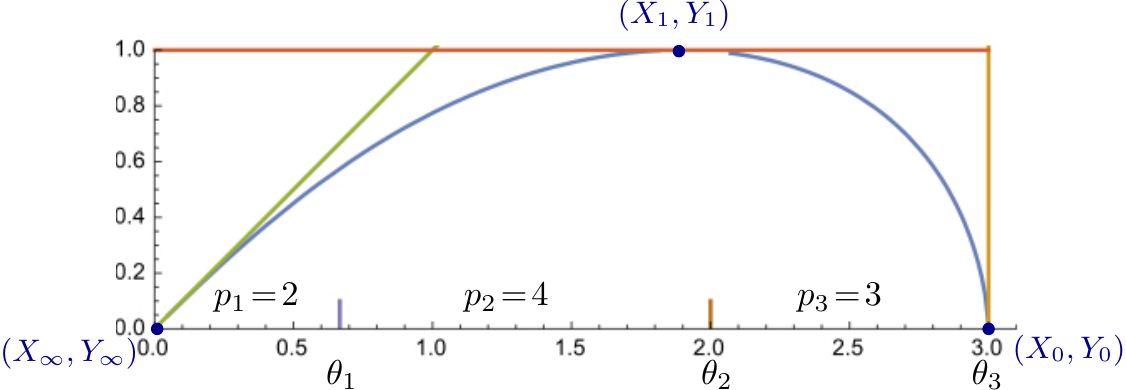}
\end{center}
\caption{The arctic curve when $\al(u)$ is continuous piecewise linear, made of $k=3$ linear pieces of respective widths $\gamma_1=\gamma_2=\gamma_3=1/3$ and slopes $p_1=2$, $p_2=4$ and $p_3=3$ (so that $\theta_1=2/3$, $\theta_2=2$ and $\theta_3=3$).}
\label{fig:p2p4p3}
\end{figure}
Figure \ref{fig:p2p4p3} presents a plot of the arctic curve in the particular case of $k=3$ linear pieces, with $\gamma_1=\gamma_2=\gamma_3=1/3$,
$p_1=2$, $p_2=4$ and $p_3=3$. 

\subsection{A first non-linear case: $\al(u)=p\, u +q\, u^2$}
In the case when $\al(u)=p\, u+q\, u^2$ with $p,q$ real numbers such that $p\geq 1$ and $q> 0$, we have by eq. \eqref{defx}:
$$ x(t)=e^{\textstyle{-\int_0^1 \frac{du}{t-p u-q u^2}}}=\left(\frac{p - 2 t +\sqrt{p^2 + 4 q t}}{ p - 2 t - \sqrt{p^2 + 4 q t}}\right)^{\frac{1}{\sqrt{p^2 + 4 q t}}} \ .$$
The special points are for $p>1$:
$$ (X_\infty,Y_\infty)=(0,0)\ , \quad (X_1,Y_1)= \left(\frac{p+1}{2}+\frac{q}{3},1\right)\ ,\quad (X_0,Y_0)=(p+q,0)\ ,$$
whereas for $p=1$ we have $(X_\infty,Y_\infty)=\left(\frac{1}{1+q},\frac{1}{1+q}\right)$.
\begin{figure}
\begin{center}
\includegraphics[width=11cm]{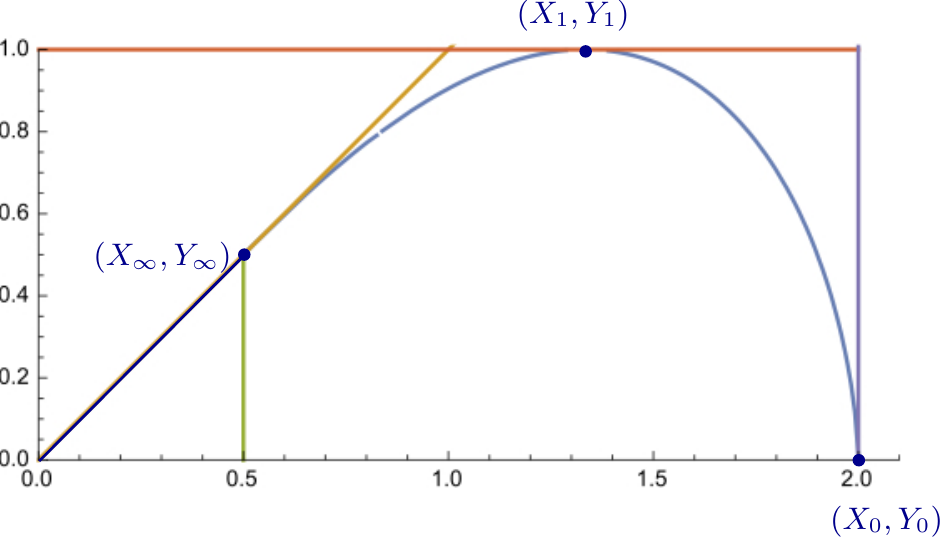}
\end{center}
\caption{The arctic curve when $\al(u)=u+u^2$ hits the $Y=X$ line at the point $(X_\infty,Y_\infty)=(1/2,1/2)$. For $X<1/2$, the limit
between the ``crystalline" and the ``liquid" phase occurs on the $Y=X$ line.}
\label{fig:nonlin}
\end{figure}
Figure \ref{fig:nonlin} presents a plot of the arctic curve in the particular case $p=q=1$. 

\subsection{A second non-linear case: $\al(u)=\frac{1}{a} u^a$}
We consider the case $\al(u)=\frac{1}{a} u^a$ for some fixed real number $a\in (0,1)$.
We have by eq. \eqref{defx}:
$$ x(t)=e^{\textstyle{-\int_0^1 \frac{du}{t-\frac{1}{a} u^a} }}=e^{\textstyle{- {}_2F_{1}\left(1,\frac{1}{a};1+\frac{1}{a}\Big\vert \frac{1}{a\, t} \right)/t}} \ ,$$
in terms of the hypergeometric function
$$ {}_2F_{1}\left(1,\frac{1}{a};1+\frac{1}{a}\Big\vert x \right)=\sum_{n\geq 0} \frac{x^n}{na +1}\ . $$
The special points are as follows: for $t\to \infty$: $(X_1,Y_1)= \left(\frac{1}{2}+\frac{1}{a(a+1)},1\right)$ with horizontal tangent. 
For $t\to 0$, we have, according to the discussion at the end of Section \ref{asymptotwo}:
$$(X_\infty,Y_\infty)= \left\{ 
\begin{split}
& \left((1-2a)\frac{e^\frac{a}{1-a}-1}{a^2},(1-2a)\frac{(e^\frac{a}{1-a}-1)^2}{a^2\,e^\frac{a}{1-a}}\right) 
\quad {\rm if}\ a<\frac{1}{2} \\
& \ \ (0,0) \qquad {\rm if}\ a\geq \frac{1}{2}\ , \\
\end{split}
\right.
$$
\begin{figure}
\begin{center}
\includegraphics[width=12cm]{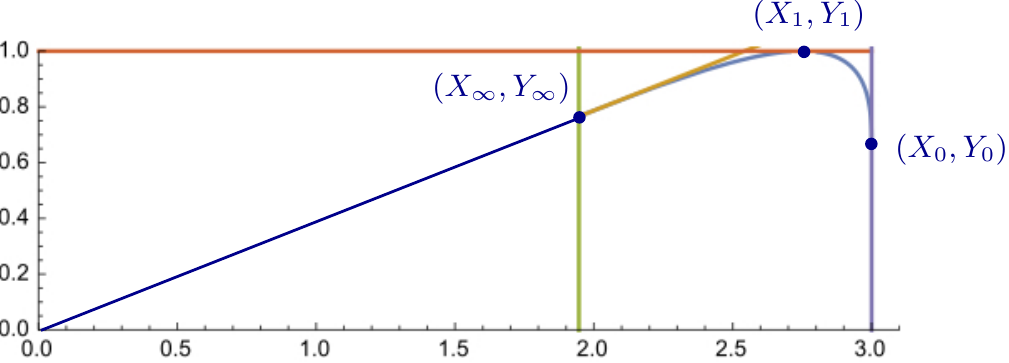}
\end{center}
\caption{The arctic curve when $\al(u)=3\, u^{1/3}$ (see text for the values of the special points). The slope at the point $(X_\infty,Y_\infty)$ is equal to 
$1-1/\sqrt{e}$. For $X<X_\infty$, the limit
between the ``crystalline" and the ``liquid" phase occurs on the line $Y=(1-1/\sqrt{e})X$.}
\label{fig:a1-3}
\end{figure}
\begin{figure}
\begin{center}
\includegraphics[width=12cm]{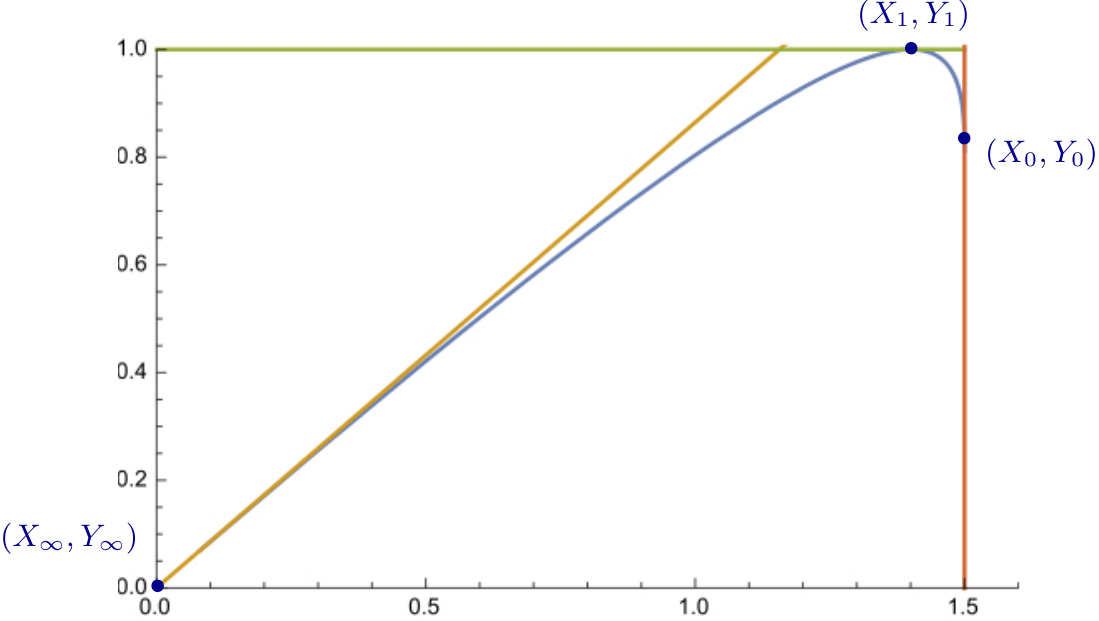}
\end{center}
\caption{The arctic curve when $\al(u)=\frac{3}{2}\, u^{2/3}$ (see text for the values of the special points). The slope at the point $(X_\infty,Y_\infty)$ is equal to 
$1-1/e^2$.}
\label{fig:a2-3}
\end{figure}
where we have used the value $J_1=\frac{a}{1-a}$ while $J_2=\frac{a^2}{1-2a}$ 
when $a<\frac{1}{2}$ and diverges otherwise. In both cases the tangent has slope $1-e^{-\frac{a}{1-a}}$.
Finally, when $t\to \al(1)=1/a$, we have $\al'(1)=1$, leading to the endpoint
$$ (X_0,Y_0)=\left(\frac{1}{a},\frac{1}{a}e^{-\gamma_E-\psi(a^{-1})}\right)
$$
by applying \eqref{eq:X0Y0}, and where $\gamma_E=.5772...$ is Euler's Gamma constant and
$\psi(u)=\Gamma'(u)/\Gamma(u)$.
We have represented the cases $a=\frac{1}{3}$ and $a=\frac{2}{3}$ in Figures \ref{fig:a1-3} and \ref{fig:a2-3} respectively. 
The special points read respectively:
\begin{eqnarray*}
&& \hskip -.6cm a=\frac{1}{3}:\ (X_1,Y_1)=\left(\frac{11}{4},1\right)\!, \ (X_\infty,Y_\infty)=\left(3(\sqrt{e}\!-\!1),3\frac{(\sqrt{e}\!-\!1)^2}{\sqrt{e}}\right)\!,
\ (X_0,Y_0)=\left(3,\frac{3}{e\sqrt{e}}\right), \\
&&\hskip -.6cm  a=\frac{2}{3}:\ (X_1,Y_1)=\left(\frac{7}{5},1\right),    \ (X_\infty,Y_\infty)=(0,0),
\ (X_0,Y_0)=\left(\frac{3}{2},\frac{6}{e^{2}}\right),
\end{eqnarray*}
with horizontal tangents at $(X_1,Y_1)$, vertical tangents at $(X_0,Y_0)$, and tangents of respective slopes
$1-1/\sqrt{e}$ and $1-1/e^2$ at $(X_\infty,Y_\infty)$.

\section{Freezing boundaries}
\label{freezesec}

\begin{figure}
\begin{center}
\includegraphics[width=14cm]{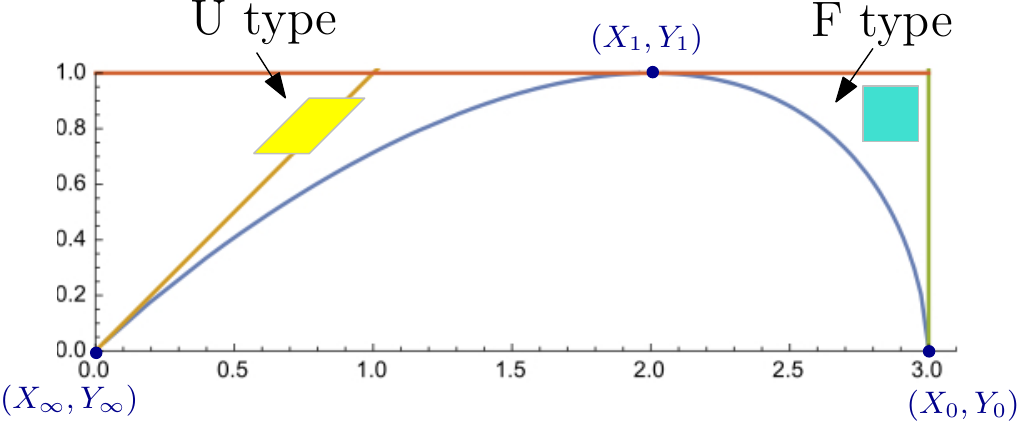}
\end{center}
\caption{Frozen domains for a generic $\al(u)$ (here in the case $\al(u)=3\, u$) are made of upper tiles (U-type) above the portion of arctic curve from $(X_\infty,Y_\infty)$ to $(X_1,Y_1)$ and made of front tiles (F-type) above the portion of arctic curve from $(X_1,Y_1)$ to $(X_0,Y_0)$.}
\label{fig:UF}
\end{figure}

So far we discussed two portions of the arctic curve, one going from $(X_\infty,Y_\infty)$ to $(X_1,Y_1)$ and 
one from $(X_1,Y_1)$ to $(X_0,Y_0)$. For a generic function $\alpha(u)$, we expect that these two portions 
build the entire arctic curve, which therefore defines two frozen domains in $\mathcal{D}$. The domain
lying above the portion from $(X_\infty,Y_\infty)$ to $(X_1,Y_1)$ corresponds in the original path family setting  
to a region where the paths are frozen into horizontal segments, or equivalently, in the second path family setting, to a region
not visited by the paths. In the tiling language, this corresponds to a frozen domain made of \emph{upper tiles}: we therefore shall refer to 
such freezing as being \emph{of type U} (for upper), see Figure \ref{fig:UF} for an illustration. As for the domain 
lying above the portion from $(X_1,Y_1)$ to $(X_0,Y_0)$, it corresponds to a region
not visited by the paths in the original path family setting  and to a region where paths of the second family form horizontal segments.
In other words, we have here a frozen domain  of type F (i.e.\ made of front tiles).
\begin{figure}
\begin{center}
\includegraphics[width=8cm]{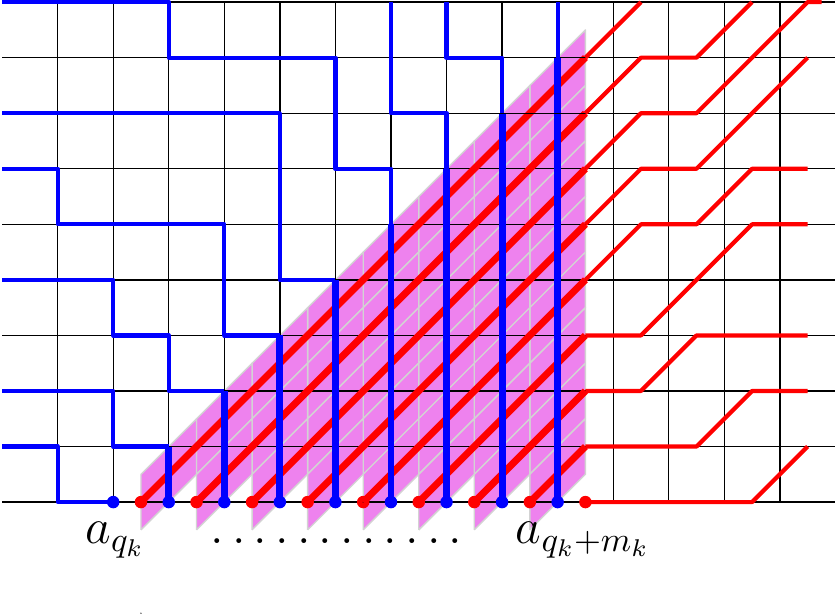}
\end{center}
\caption{A schematic picture of a freezing boundary, where $a_{i+1}-a_i=1$ for $i$ in some interval $I_k=\{q_k,q_k+1,\dots,q_k+m_k-1\}$. 
The non-intersection constraint creates a fully frozen triangular region made of right tiles only. 
This region will serve as a germ for a larger frozen domain of type R around it.}
\label{fig:freezing}
\end{figure}
A frozen domain with the third possible type of freezing, of type R (i.e. made of right tiles with paths of the first family frozen vertically, or equivalently, paths of the second family frozen along diagonal lines)
will not appear in general since for a generic increasing sequence, the spacing between the successive $a_i$'s leaves enough space
for the paths to develop some fluid erratic behavior \emph{in the horizontal direction}. 

New portions of arctic curve may still appear in the presence of what may be called \emph{freezing boundaries},
i.e.\ for particular sequences $(a_i)_{0\leq i\leq n}$ which induce new frozen domains adjacent to the lower boundary of the domain $D$.

A first kind of such freezing boundary corresponds to a case for which there is no (horizontal) spacing left in-between successive $a_i$'s. 
In other words, it may happen that $a_{i+1}-a_i=1$ for $i$ lying in one or several "macroscopic" intervals $I_k=\{q_k,q_k+1,\dots,q_k+m_k-1\}$ where the length $m_k$ of $I_k$ 
scales like $n$. As displayed in Figure \ref{fig:freezing}, the non-intersection constraint in this case creates, for any
such interval, a triangular region which is fully frozen, of type $R$. 
We expect these fully frozen regions to then serve as germs for even larger frozen domains of type R, 
hence to create new portions for the arctic curve. Note that, for the third family of paths made of north- and northeast-oriented steps, 
these freezing domains of type R correspond to regions not visited by the paths.

The condition that $a_{i+1}-a_i=1$ for $i\in I_k$ translates into the condition $\alpha'(u)=1$ for $u$ in some finite interval $[u_k,u_k+\gamma_k]$
(with $u_k=q_k/n$ and $\gamma_k=m_k/n>0$ in the large $n$ limit). When several intervals co-exist, they may be arranged into a family of (maximal) disjoint intervals 
$[u_k,u_k+\gamma_k]$ (where $u_{k+1}>u_k+\gamma_k$), which may possibly include boundary intervals of the form $[0,\gamma]$ or 
$[1-\gamma,1]$.
 
\begin{figure}
\begin{center}
\includegraphics[width=8cm]{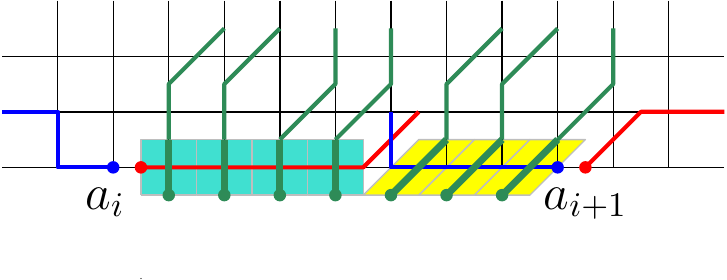}
\end{center}
\caption{A schematic picture of a freezing boundary corresponding to a ``macroscopic" gap in the 
sequence between $a_i$ and $a_{i+1}$ for some $i$. This forces the lower layer to be made of a sequence of front tiles 
followed by a sequence of upper tiles.  This frozen layer will serve as germ for two extended frozen domains: one of type F above the left part
and one of type U above the right part.}
\label{fig:freezingtwo}
\end{figure}
\medskip
Another type of freezing boundary corresponds to the opposite case where there is one or several "macroscopic" gaps in the 
sequence $(a_i)_{0\leq i\leq n}$, namely intervals $I_k=\{q_k,q_k+1,\dots,q_k+m_k-1\}$ (with $m_k$ scaling like $n$) which contain
no $a_i $ at all. As displayed in Figure \ref{fig:freezingtwo}, this case creates, for any
such interval, a fully frozen layer made of a sequence of front tiles followed by a sequence of upper tiles (so that the lower boundary 
of the layer is horizontal). We expect these frozen layers to serve as germs for extended frozen domains of type F above their left part
and of type U above their right part, creating again new portions for the arctic curve.

The presence of gaps translates into the fact that $\alpha(u)$ is discontinuous and presents a jumps of height $\delta_k=m_k/n$ at $u_k=q_k/n$.

This section is devoted to a heuristic study of these freezing boundaries, of both types, creating new portions of arctic curve.

\subsection{The case of a piecewise linear $\al(u)$ revisited}\label{revisited}
We may easily introduce freezing boundaries in the framework studied in Section \ref{sec:piecewise} where $\al(u)$ is a continuous and piecewise 
linear function made of $k$ pieces, as defined in Section \ref{sec:piecewise}. Let us start with freezing boundaries creating frozen domains of type R.
Such boundaries exist whenever $p_i=1$ for one or several $i$'s in $\{1,\dots,k\}$. 

To describe new portions of the arctic curve, we note that, in all generality, the two already known portions are described by 
\emph{the same parametric equations}, given by \eqref{arcticone} or \eqref{arctictwo} with the same expression \eqref{defx} for $x(t)$.
Only the range of $t$ differs between the two portions, namely $(-\infty,0]$ for one portion and $[\al(1),+\infty)$ for the other.
This range covers the allowed $X$-coordinates of the points at which the tangents intersect the $X$-axis, whose value is precisely $t$. 
The allowed values of $t$ correspond moreover to positive real values of $x(t)$
ranging from $0$ to $\infty$, the slope of the tangent parametrized by $t$ being precisely $-(1-x(t))/x(t)$.

 It is tempting to conjecture that, in the presence of freezing boundaries, the expected new portions of the arctic curve are again given by \eqref{arcticone} (or \eqref{arctictwo}) and simply correspond
to new possible values of the parameter $t$. In order for these parametric equations to remain meaningful, we must insist on having a \emph{real value}
for $x(t)$. On the other hand, releasing the constraint that $x(t)$ be positive seems harmless. Let us now see how this may be realized in the piecewise linear case.

From  the expression \eqref{piecesol} for $x(t)$, written as
\begin{equation*}
 x(t)=\prod_{i=1}^k \left(\frac{t-\theta_{i}}{t-\theta_{i-1}}\right)^{\frac{1}{p_i}}\ ,
  \end{equation*}
we immediately see that the $i$-th term in the product leads to a cut of $x(t)$ on the real interval $[\theta_{i-1},\theta_i]$ when $p_i>1$.
If all the $p_i$'s are strictly larger than $1$, then $x(t)$ has a cut on the real axis along the whole interval $[0,\theta_k]$ and, for real $t$, is
well-defined only for $t\geq \theta_k=\al(1)$ or $t\leq 0$ (for which $x(t)$ is moreover real and positive) corresponding
to the known two portions of the arctic curve. On the other hand, if $p_m=1$ for some $m$, 
then the above formula is well defined on $[\theta_{m-1},\theta_m]$ and takes the value:
\begin{equation*}
  x(t) =- \prod_{i=1}^{m-1} \left(\frac{t-\theta_{i}}{t-\theta_{i-1}}\right)^{\frac{1}{p_i}}\times \left(\frac{\theta_{m}-t}{t-\theta_{m-1}}\right)\times \prod_{i=m+1}^k \left(\frac{\theta_{i}-t}{\theta_{i-1}-t}\right)^{\frac{1}{p_i}}\ \  \hbox{for}\ \  t\in [\theta_{m-1},\theta_m]\  \ (p_m=1)\ .
 \end{equation*}
Taking $p_m=1$ therefore gives rise to a domain $[\theta_{m-1},\theta_m]$ of $t$ for which $x(t)$ is real and negative.
This new range of $t$ in turn gives rise via the equation \eqref{eq:tgeq} (or equivalently \eqref{eq:tgeqtwo}) to a new set of 
tangent lines with positive slope $-(1-x(t))/x(t)$ crossing the $X$-axis at $(t,0)$ with $t\in [\theta_{m-1},\theta_m]$, which is precisely the location of the base of the triangular fully frozen region of type R (as displayed in 
Figure \ref{fig:freezing}).  It is easily checked that the slope of the tangent is equal to $1$ for $t=\theta_{m-1}$ and $\infty$ for $t=\theta_m$
and that the envelope of these tangents for $t\in [\theta_{m-1},\theta_m]$ presents a cusp. We conjecture that this envelope is precisely 
the outer boundary of a larger frozen domain
enclosing the fully frozen triangular region and  tangent to this region at its endpoints. We thus have here a new portion of arctic curve.

\begin{figure}
\begin{center}
\includegraphics[width=10cm]{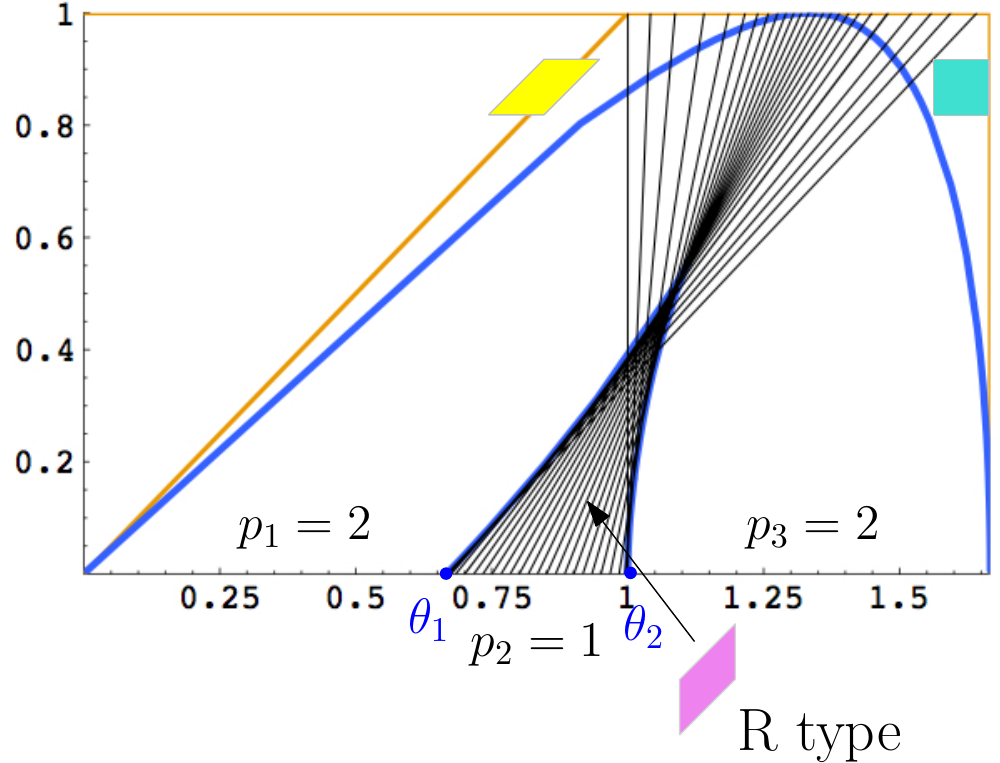}
\end{center}
\caption{The complete arctic curve when $\al(u)$ is continuous piecewise linear with $k=3$ linear pieces of respective widths $\gamma_1=\gamma_2=\gamma_3=1/3$
and slopes $p_1=2$, $p_2=1$ and $p_3=2$ (so that $\theta_1=2/3$ and $\theta_2=1$), giving rise to a freezing boundary along the segment $[\theta_1,\theta_2]$. A new frozen domain of type R emerges above this segment, separated from the ``liquid" phase by a new portion of arctic curve
forming a cusp. As displayed, this new portion is obtained as the envelope of a family of tangents whose intercepts with the $X$-axis have abscissa $t\in [\theta_1,\theta_2]$.}
\label{fig:reentrance}
\end{figure}

Figure \ref{fig:reentrance} displays for illustration the complete (including conjectured portions) arctic curve in the case $k=3$, $\gamma_1=\gamma_2=\gamma_3=1/3$, 
$p_1=p_3=2$ and $p_2=1$.  Clearly, when $p_m=1$ for several values of $m$ (which we take non consecutive as, in the piecewise linear setting, it is 
implicitly assumed that consecutive slopes are different), each piece where $p_m=1$ gives rise to a new frozen domain. 
When a freezing boundary occurs in the first piece (i.e.\ when $p_1=1$), it is easily checked that $Y_\infty>0$ and that the new frozen domain
is enclosed by a new portion of arctic curve from $(X_\infty,Y_\infty)$ to $(\theta_1,0)$. Similarly, when a freezing boundary occurs in the last piece 
(i.e.\ when $p_k=1$), the new frozen domain
is enclosed by a new portion of arctic curve from $(X_0,Y_0)$ (where $Y_0> 0$) to $(\theta_{k-1},0)$. Figure \ref{fig:p1p2p1p2p1} displays
a situation where both $p_1$ and $p_k$ are equal to $1$, namely the case $k=5$, $p_1=p_3=p_5=1$, $p_2=p_4=2$ and $\gamma_i=1/5$ for
$i=1,\dots,5$, giving rise to three new frozen domains.
  
\begin{figure}
\begin{center}
\includegraphics[width=12cm]{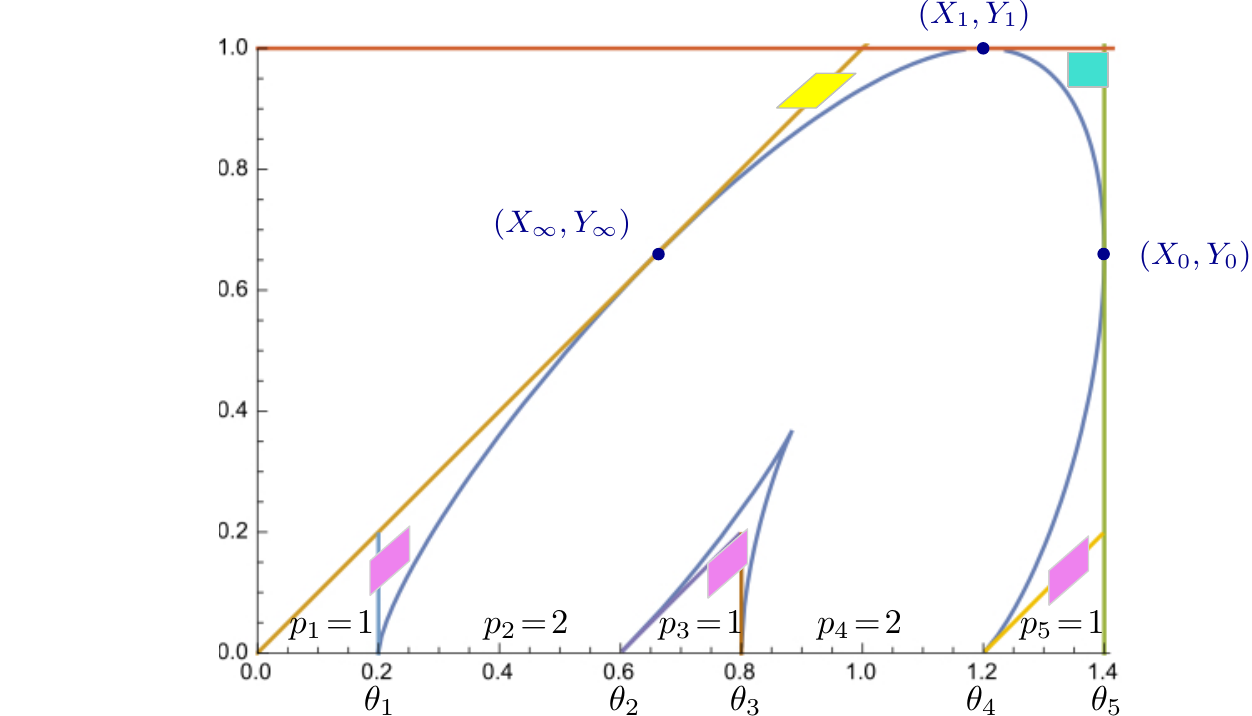}
\end{center}
\caption{The complete arctic curve when $\al(u)$ is continuous piecewise linear, made of $k=5$ linear pieces of widths $\gamma_i=1/5$ ($i=1,\dots,5$) and slopes $p_1=p_3=p_5=1$, $p_2=p_4=2$. The three pieces where $p_i=1$ give rise to freezing boundaries which generate frozen domains of type R.
The first frozen domain of type R is separated from the ``liquid" phase by a portion of arctic curve joining $(\theta_1,0)$ to $(X_\infty,Y_\infty)$ and from 
a frozen domain of type U by the $Y=X$ line for $X\leq X_\infty$. The third frozen domain of type R is separated from the ``liquid" phase by a portion of arctic curve joining $(\theta_4,0)$ to $(X_0,Y_0)$. For each frozen domain we indicated its triangular ``germ".  }
\label{fig:p1p2p1p2p1}
\end{figure}

\medskip
Let us now come to the case of freezing boundaries arising from a gap in the $a_i$'s, creating frozen domains of type F and U. This situation also \
occurs in the setting of piecewise linear functions $\al(u)$ in the following limit. A discontinuity in the function $\al(u)$ may be obtained 
by letting  $\gamma_m\to 0$ for some $m$ together with $p_m\to \infty$, keeping the product $p_m \gamma_m=\delta_m$ finite. This creates 
a jump in the function $\al(u)$ by $\delta_m$ at the position $u=\varphi_{m-1}=\varphi_m$ (recall that $\varphi_i:=\sum_{j=1}^i \gamma_j$).
Using again the parameters $\theta_i:=\sum_{j=1}^i p_j\gamma_j$ to express $x(t)$, we have the identification $\delta_m=\theta_m-\theta_{m-1}$
so that we may use the form \eqref{piecesol} for $x(t)$, now with $p_m=\infty$ to write
\begin{equation*}
  x(t) = \prod_{i=1}^{m-1} \left(\frac{t-\theta_{i}}{t-\theta_{i-1}}\right)^{\frac{1}{p_i}}\times \prod_{i=m+1}^k \left(\frac{\theta_{i}-t}{\theta_{i-1}-t}\right)^{\frac{1}{p_i}}\ \ \ (p_m=\infty)\ .
 \end{equation*}
\begin{figure}
\begin{center}
\includegraphics[width=10cm]{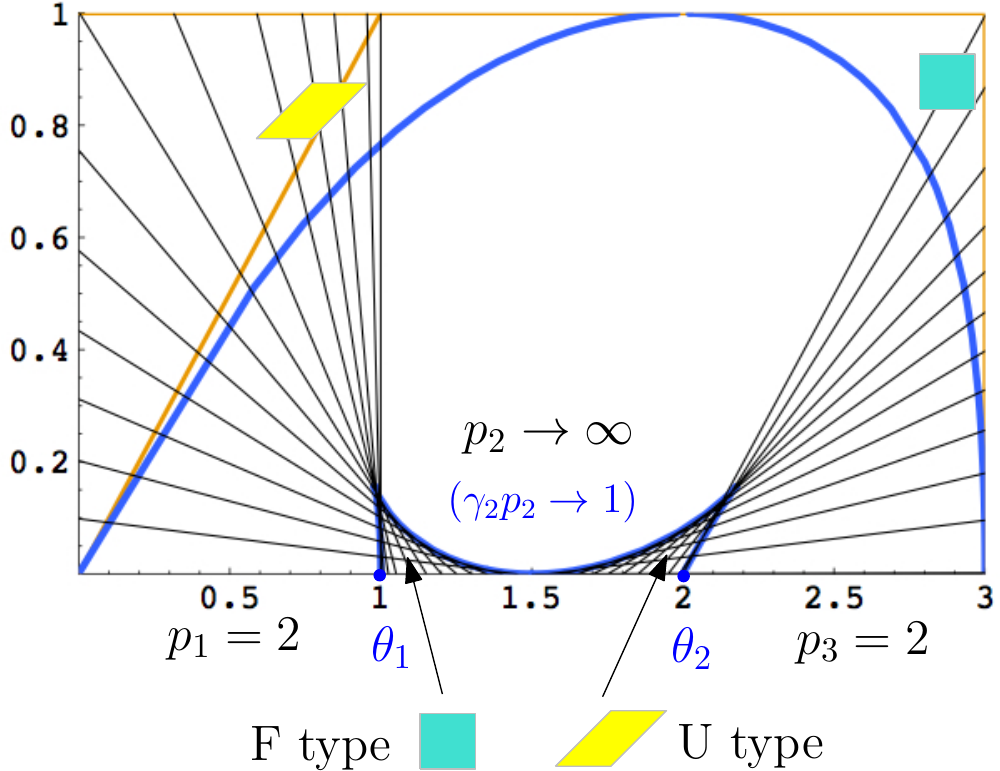}
\end{center}
\caption{The complete arctic curve when $\al(u)$ is a piecewise linear function made of two pieces with a discontinuity in-between, 
obtained as limit of a continuous piecewise linear function made of $k=3$ pieces of widths $\gamma_1=\gamma_3=1/2$, $\gamma_2\to 0$, and slopes
$p_1=p_3=2$, $p_2\to \infty$,  with $p_2\gamma_2\to \delta_2=1$ (so that $\theta_1=1$ and $\theta_2=\theta_1+\delta_2=2$). The discontinuity gives rise to a freezing boundary with a gap along the segment $[\theta_1,\theta_2]$. Two new frozen domains of respective type F and U emerge 
above this segment, separated from the ``liquid" phase by a new portion of arctic curve
forming two cusps and being tangent to the $X$-axis at some point with abscissa between $\theta_1$ and $\theta_2$ (here equal to $3/2$ by symmetry). As displayed, this new portion is obtained as the envelope of a family of tangents whose intercepts with the $X$-axis have abscissa $t\in [\theta_1,\theta_2]$. For clarity, the $Y$-axis has been stretched.  }
\label{fig:reentrancegap}
\end{figure}
Apart from the domains $t\leq 0$ and $t\geq \al(1)=\theta_k$, this opens a new domain $[\theta_{m-1},\theta_m]$ of linear size $\delta_m$
for the allowed values of $t$, leading to real and positive values of $x(t)$. As displayed in
Figure \ref{fig:reentrancegap} (which shows the resulting complete arctic curve in the simple case $k=3$, $\gamma_1=\gamma_3=1/2$, 
$p_1=p_3=2$ and $p_2\to \infty$, $\gamma_2\to 0$, $p_2\gamma_2\to \delta_2=1$),  the corresponding family of tangents
creates a new portion of arctic curve made of three parts: a part on the left leaving the point $(\theta_{m-1},0)$ with a vertical slope,  a part on the right 
leaving the point $(\theta_m,0)$ with a slope $1$  and a middle part which is tangent to the $X$-axis at a point $(\theta,0)$ for some 
$\theta\in [\theta_{m-1},\theta_m]$. This in turn creates two frozen domains, one of type F on the left, and one of type U on the right.

\subsection{Freezing the right edge: exact derivation}

So far, the expressions for the new portions that we obtained are based on the conjecture that the parametric equation 
for the arctic curve is not only valid for $t$ in the range $(-\infty,0]\cup[\al(1),+\infty)$ but holds in a larger range of values 
corresponding to real values of $x(t)$. This hypothesis may be tested in the particular case where the freezing boundary 
lies on the edge of the domain $D$. More precisely, this section is devoted to the study of the effect of ``freezing the right edge" of our paths 
by imposing that the rightmost starting points obey
$a_{i+1}-a_{i}=1$ for $i=n-r+1,n-r+2,\dots,n-1$, while $a_{n-r+1}-a_{n-r}>1$, 
and letting $r$ grow proportionally to $n$ when $n$ becomes large. In turn, letting $r=\rho n$, this amounts to the condition 
$\al'(u)=1$ on the segment $[1-\rho,1]$. We expect in this case a frozen domain of type R below a new portion of arctic curve 
connecting the point $(\al(1-\rho),0)$ to the point $(X_0,Y_0)$ (where $Y_0>0$ in this case). Let us show that this is indeed the case.

\subsubsection{Partition function: a new derivation}

It is easier to describe the present situation in terms of the complementary
starting points $b_i$, $i=1,2,\dots,m$, for the paths with north- and northeast-oriented steps of Section \ref{pf}, where $m+n=a_n$. 
The above condition simply forces the position $b_m=a_n-r$ of the rightmost starting point.
As mentioned in Section \ref{pf}, the partition function for paths with north- and northeast-oriented steps, starting at $(b_i,-1/2)$, $i=1,2,\dots,m$ and ending at 
$(n+j,n+1/2)$, $j=1,2,\dots,m$ is given by the determinant of the LGV matrix ${\hat A}_{i,j}$ with entries:
\begin{equation}\label{hatAlgv}
{\hat A}_{i,j}={n+1\choose b_i-j+1} \qquad (i,j=1,2,\dots,m)\ .
\end{equation}
Let us use again the LU decomposition method to compute the determinant directly in terms of the $b$'s. We have the following explicit result:

\begin{thm}
The lower uni-triangular matrix ${\hat L}^{-1}$ with elements:
\begin{equation}
{\hat L}^{-1}_{i,j}=\left\{
\begin{matrix}\frac{\displaystyle {n+m\choose b_i} {n+m-b_i\choose m+1-i}}{\displaystyle {n+m\choose b_j}{n+m-b_j\choose m+1-i}} \ 
\frac{\displaystyle \prod_{s=1}^{i-1} (b_i-b_s)}{\displaystyle \prod_{s=1\atop s\neq j}^i (b_j-b_s)}
& \hbox{for}\ i\geq j\\
0 & \hbox{for}\ i<j
\end{matrix}
\right.
\end{equation}
is such that ${\hat U}:= {\hat L}^{-1}{\hat A}$ is upper triangular.
\end{thm}
\begin{proof}
We compute:
\begin{eqnarray*}
{\hat U}_{i,j}&=& \sum_{k=1}^m ({\hat L}^{-1})_{i,k} {n+1\choose b_k-j+1}\\
&=& \sum_{k=1}^i\frac{\displaystyle {n+m\choose b_i} {n+m-b_i\choose m+1-i}}{\displaystyle {n+m\choose b_k}{n+m-b_k\choose m+1-i}} \ 
\frac{\displaystyle \prod_{s=1}^{i-1} (b_i-b_s)}{\displaystyle \prod_{s=1\atop s\neq k}^i (b_k-b_s)} \, {n+1\choose b_k-j+1}\ .\\
\end{eqnarray*}
Note that, due to the binomial factors, only the values of $k$ for which $j-1\leq b_k\leq n+j$ and $b_k\leq n+i-1$ contribute to the sum.
When this holds, the combination of the five binomial factors above may then be rewritten as
\begin{equation*}
\left\{\begin{split}
&\frac{(n+1)!}{b_i!(n+i-1-b_i)!}\ \prod_{s=0}^{j-2}(b_k-s) \prod_{s=j+1}^{i-1}(n-b_k+s)\quad \hbox{for}\ i>j\\
&\frac{(n+1)!}{b_i!(n+i-1-b_i)!}\ \frac{\displaystyle{\prod_{s=0}^{j-2}(b_k-s)}}{\displaystyle {\prod_{s=i}^{j}(n-b_k+s)}}\quad \hbox{for}\ i\leq j\ .\\
\end{split}\right.
\end{equation*}
Assume now that $i>j$ so that the constraint over $b_k$ reduces to $j-1\leq b_k\leq n+j$. We way then write 
$$ {\hat U}_{i,j}= \frac{(n+1)!}{b_i!(n+i-1-b_i)!} \prod_{s=1}^{i-1} (b_i-b_s) 
\oint_{{\mathcal C}(\hat{S}_{j})} \frac{dt}{2{\rm i}\pi}
\prod_{s=0}^{j-2}(t-s)\, \prod_{s=j+1}^{i-1} (n-t+s) \, \prod_{s=1}^i \frac{1}{(t-b_s)}
$$
where the contour encompasses only the set $\hat{S}_j=\{b_s \vert  j-1\leq b_s\leq n+j\}$.

Due to the factor $\prod_{s=j+1}^{i-1} (n-t+s)$ which vanishes for $t=n+j+1,n+j+2,\dots,n+i-1$ and to the factor $\prod_{s=0}^{j-2}(t-s)$ which 
vanishes for $t=0,1,\dots, j-2$, the contour of integration 
can be extended harmlessly so as to encircle all the poles $b_1,b_2,\dots,b_i$ as the residues of the unwanted contributions
vanish (recall that $b_i\leq n+i-1$ since $b_i<b_{i+1}<\dots<b_m<n+m$). 
In turn, by the Cauchy theorem, the integral can be expressed as minus the contribution of the pole at $\infty$. But for large $t$,
the integrand behaves as $t^{-2}$, hence the residue at $\infty$ vanishes, and we conclude that ${\hat U}_{i,j}=0$ for $i>j$, i.e. $\hat U$
is upper triangular.
\end{proof}

The diagonal matrix elements ${\hat U}_{i,i}$ are also easily obtained from the above:
\begin{equation*} {\hat U}_{i,i}=
 \frac{(n+1)!}{b_i!(n+i-1-b_i)!} \prod_{s=1}^{i-1} (b_i-b_s) 
\oint_{{\mathcal C}(b_1,b_2,\dots,b_i)} \frac{dt}{2{\rm i}\pi}\,\prod_{s=1}^i \frac{1}{(t-b_s)}
\, \frac{\displaystyle\prod_{s=0}^{i-2}(t-s)}{(n+i-t)}
\end{equation*}
where the contour encompasses all $b_s$ for $s=1,2,\dots,i$, but not $n+i$. Indeed the original contour must select those $b_s$ with
$i-1\leq b_s\leq n+i-1$ and may be extended to those $b_s$ with $0\leq b_s\leq n+i-1$ (due to the vanishing of $\prod_{s=0}^{i-2}(t-s)$ for $t=0,1,\dots, i-2$),
which includes all $b_s$ for $s=1,2,\dots,i$  since the condition $b_s\leq n+i-1$ is automatically satisfied (due to $b_i\leq n+i-1$). 
As before we note that the integrand behaves as $1/t^2$
for large $t$, hence the residue at $\infty$ vanishes. By the Cauchy theorem, we may therefore re-express $ {\hat U}_{i,i}$ as minus
the residue at the excluded pole $n+i$. We find:
\begin{eqnarray*} {\hat U}_{i,i}&=&
 \frac{(n+1)!}{b_i!(n+i-1-b_i)!} \prod_{s=1}^{i-1} (b_i-b_s) 
\, \prod_{s=1}^i \frac{1}{(n+i-b_s)}
\, \prod_{s=0}^{i-2}(n+i-s)\\
&=& {n+i\choose b_i} \prod_{s=1}^{i-1} \frac{(b_i-b_s)}{(n+i-b_s)}\ .
\end{eqnarray*}
This leads to the following result:
\begin{thm}
The partition function expressed in terms of the sequence $(b_i)_{1\leq i\leq m}$ reads:
\begin{equation*}
Z_n=\frac{\Delta(0,1,\dots,n+m)}{\Delta(0,1,\dots,n)}\,
\frac{\displaystyle \Delta(b_1,b_2,\dots,b_m)}{\displaystyle\prod_{i=1}^m b_i! (n+m-b_i)!}\ .\label{Zofb}
\end{equation*}
\end{thm}
\begin{proof}
We compute
\begin{eqnarray*}
Z_n&=&\prod_{i=1}^m {\hat U}_{i,i}=\prod_{i=1}^m \left\{ {n+i\choose b_i} \prod_{s=1}^{i-1} \frac{(b_i-b_s)}{(n+i-b_s)}\right\}=\prod_{i=1}^m  
\left\{\frac{(n+i)!}{b_i! (n+m-b_i)!}\,\prod_{s=1}^{i-1}(b_i-b_s)\right\} \\
&=&\frac{\displaystyle\prod_{i=1}^{m+n} \left\{ i!\ \displaystyle\prod_{1\leq i<j \leq m}^{i-1}(b_j-b_i)\right\}}{\displaystyle\prod_{i=1}^{n}
\left\{ i!\ \displaystyle\prod_{i=1}^m b_i! (n+m-b_i)!\right\}}\,
 =\frac{\Delta(0,1,\dots,n+m)}{\Delta(0,1,\dots,n)}\,
\frac{\displaystyle \Delta(b_1,b_2,\dots,b_m)}{\displaystyle\prod_{i=1}^m b_i! (n+m-b_i)!}\ .
\end{eqnarray*}
\end{proof}
Note that this evaluation of the determinant of the matrix $\hat A$ of \eqref{hatAlgv} is a particular limit
$q\to 1$ of a more general formula \cite{KRADET}, Theorem 26, eq. (3.12). 

Using the complementarity of the $a$'s and $b$'s, namely $\{a_s\}\cup \{b_q\}=\{0,1,\dots,n+m\}$ and  $\{a_s\}\cap \{b_q\}=\emptyset$, 
we have the identity
\begin{eqnarray*} \Delta(0,1,\dots,n+m)&=&
\Delta(a_0,a_1,\dots,a_n)\Delta(b_1,b_2,\dots,b_m)\prod_{a_s<b_q}(b_q-a_s)\prod_{b_q<a_s}(a_s-b_q)\\
&=&\frac{\Delta(a_0,a_1,\dots,a_n)}{\Delta(b_1,b_2,\dots,b_m)} \prod_{q=1}^m b_q! (n+m-b_q)!
\end{eqnarray*}
which allows to identify \eqref{Zofb} with \eqref{paf}.

\subsubsection{One-point function}
\begin{figure}
\begin{center}
\includegraphics[width=12cm]{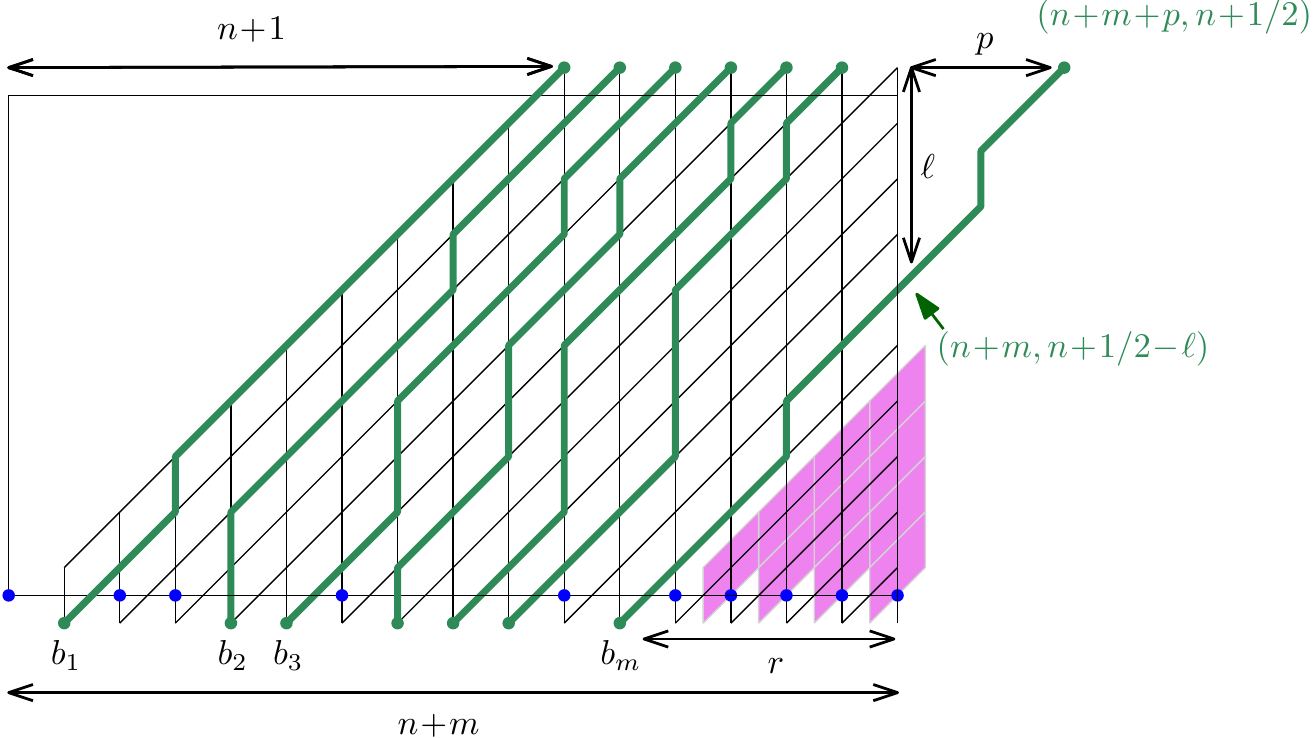}
\end{center}
\caption{The tangent method applied to NILP made of north- and northeast-oriented step paths with $b_m=a_n-r=m+n-r$, i.e.\ with a freezing boundary 
of linear size $r$ on the right of the lower boundary of the domain $D$ (displayed here as a rectangle) creating a frozen triangular region made of right tiles only. Moving the endpoint of the rightmost path from $(n+m,n+1/2)$ to $(n+m+p,n+1/2)$ with $p\in \Z_+$  forces the rightmost path to escape 
from the domain $D$ at some point $(n+m,n+1/2-\ell)$.}
\label{fig:Hnml}
\end{figure}
Let us now apply the tangent method to the configurations of north- and northeast-oriented step paths with the frozen boundary 
$b_m=a_n-r=m+n-r$, by moving the 
endpoint of the rightmost path from $(n+m,n+1/2)$ to another point on the right $(n+m+p,n+1/2)$, $p\geq 0$. This induces an escape of the rightmost path
from the domain $D$ at a point $(n+m,n+1/2-\ell)$ (see Figure \ref{fig:Hnml} for an illustration). As usual, the corresponding one-point function reads:
${\hat H}_{n,m,\ell}=\frac{ {\hat U}'_{m,m}}{{\hat U}_{m,m}}$,
where
${\hat U'}={\hat L}^{-1} {\hat A}'$, ${\hat A}'$ the LGV matrix for the configurations with an escaping path, with entries:
\begin{equation}
{\hat A}'_{i,j}=\left\{ \begin{split}
& {\hat A}_{i,j} \qquad\qquad\qquad {\rm for}\ 1\leq j<m \\
& {n-\ell+1\choose n+m-b_i} \ \quad {\rm for}\ j=m \ .
\end{split}
\right.
\end{equation}

\begin{thm}
The one-point function ${\hat H}_{n,m,\ell}$ reads
\begin{equation}\label{hatH}
{\hat H}_{n,m,\ell}=\frac{\displaystyle \prod_{s=1}^m (n+m-b_s)}{\displaystyle {n+m\choose n-\ell+1}}
\oint_{{\mathcal C}(b_1,b_2,\dots,b_m)} \frac{dt}{2{\rm i}\pi}
\frac{1}{(m\!+\!n\!-\!t)} \, \prod_{s=1}^m\frac{1}{(t-b_s)} \, \frac{\displaystyle \prod_{s=0}^{m\!+\!\ell\!-\!2}(t-s)}{\displaystyle  (m\!+\!\ell\!-\!1)!}\ ,
\end{equation}
where the contour leaves the point $m+n$ out.
\end{thm}
\begin{proof}
We compute
\begin{eqnarray*}
{\hat H}_{n,m,\ell}&=&\frac{1}{{\hat U}_{m,m}}\sum_{k=1}^m ({\hat L}^{-1})_{m,k} {n-\ell+1\choose n+m-b_k}\\
&=& \prod_{s=1}^{m}(n+m-b_s) \sum_{k=1}^m \frac{\displaystyle {n-\ell+1 \choose n+m-b_k}}{\displaystyle {n+m\choose b_k}}
\frac{1}{(n+m-b_k)\ \displaystyle{ \prod_{s=1\atop s\neq k}^m {(b_k-b_s)}}}\ ,\\
\end{eqnarray*}
where only those values of $k$ for which $b_k\geq m+\ell-1$ contribute to the sum (recall also  that $b_k< a_n=m+n$ for all $k$) .  Using
\begin{equation*}
\frac{\displaystyle {n-\ell+1 \choose n+m-b_k}}{\displaystyle {n+m\choose b_k}}=\frac{1}{\displaystyle{{n+m\choose n-\ell+1}}}\frac{\prod
\limits_{s=0}^{m+\ell-2}(b_k-s)}{(\ell+m-1)!}\ ,
\end{equation*}
we may thus write
\begin{equation*}
{\hat H}_{n,m,\ell}=\frac{\displaystyle \prod_{s=1}^m (n+m-b_s)}{\displaystyle {n+m\choose n-\ell+1}} \oint_{{\mathcal C}(\{b_s\vert b_s\geq m+\ell-1\})} \frac{dt}{2{\rm i}\pi}\frac{1}{(n+m-t)} \, \prod_{s=1}^m\frac{1}{(t-b_s)} \, \frac{\displaystyle \prod_{s=0}^{m+\ell-2}(t-s)}{\displaystyle  (m+\ell-1)!}
\ .
\end{equation*}
We may harmlessly extend the integral contour
so as to encompass all the $b_s$, as all the extra poles at $b_s<m-\ell-1$ have vanishing residues 
(due to the factor $\prod_{s=0}^{m+\ell-2}(t-s)$), and the formula \eqref{hatH} follows.
\end{proof}

The partition function for the single path from the escape point $(n+m,n+1-\ell)$, starting with a northeast-oriented step,
and ending at the target point $(n+m+p,n+1)$ is simply
\begin{equation}\label{hatY}
{\hat Y}_{p,\ell}={\ell-1\choose p-1}\ .
\end{equation}
Note in particular the condition $\ell\geq p$ (which is saturated only if all steps taken by the path are of the
northeast type).

\subsubsection{Asymptotic analysis}

For large $n$, we use the scaling $m=\mu n$, $r=\rho n$, $p=w n$, $\ell=\hat\xi n$, 
and $b_i=n\beta(i/n)$ with a piecewise differentiable function $\beta(u)$
with $\beta'(u)\geq 1$ when defined. Moreover the freezing condition implies that $b_m=a_n-r=n+m-r$, hence $\beta(\mu)=1+\mu-\rho$,
with $\rho>0$.

\begin{thm}\label{arcticthreethm}
The tangent method for the case of a target endpoint to the east of $D$ and an escape point on the right boundary of $D$, leads to the following portion of arctic curve:
\begin{equation}\label{arcticthree}
\left\{ \begin{matrix}
X=X(t):= & t- \frac{\displaystyle{x(t)(1-x(t))}}{\displaystyle{x'(t)}} \hfill\\
& \\
Y=Y(t):= &\frac{\displaystyle{(1-x(t))^2}}{\displaystyle{x'(t)}}\hfill
\end{matrix}\right. \qquad (t\in [1+\mu-\rho,1+\mu])\ ,
\end{equation}
with 
\begin{equation}\label{defy}
x(t)= -\frac{1+\mu-t}{t\, y(t)}\, ,\quad y(t)=e^{-\textstyle\int_0^\mu \frac{du}{t-\beta(u)}}\ . 
\end{equation}
\end{thm}
\begin{proof}
From the explicit expressions \eqref{hatH} and \eqref{hatY} for ${\hat H}_{n,m,\ell}$
and ${\hat Y}_{p,\ell}$, we may infer the scaling limits
$$ {\hat H}_{n,\mu n,\hat\xi  n}\sim \oint 
\frac{dt}{2{\rm i}\pi} e^{n{\hat S}_0(t,\hat\xi)} ,\quad {\hat Y}_{w n,\hat\xi n}\sim e^{n{\hat S}_1(\hat\xi)} \ ,$$
where we have performed the customary redefinition $t\to n t$,. The contour, which, before rescaling, encompasses all the $b_i$'s but leaves the point $(n+m)$ out, must encircle the real segment $[0,1+\mu-\rho]$ but \emph{leave
the point $(1+\mu)$ out}, i.e.\  cross the real axis strictly inside the segment $[1+\mu-\rho,1+\mu]$ as well as on the negative real axis $(-\infty,0]$. 
Here we have
\begin{eqnarray*}
{\hat S}_0(t,\hat\xi)&=&\int_0^\mu{\rm Log}\left( \frac{1+\mu-\beta(u)}{t-\beta(u)}\right)+t{\rm Log}(t)-(t-\mu-\hat\xi){\rm Log}(t-\mu-\hat\xi)\\
&&-(1+\mu){\rm Log}(1+\mu)+(1-\hat\xi){\rm Log}(1-\hat\xi) \\
{\hat S}_1(\hat\xi)&=& \hat\xi{\rm Log}(\hat\xi)- w {\rm Log}(w)-(\hat\xi-w){\rm Log}(\hat\xi-w)\ .
\end{eqnarray*}
The saddle-point and maximum equations $\partial_t {\hat S}_0=\partial_{\hat\xi} ({\hat S}_0+{\hat S}_1)=0$ lead to
\begin{equation*}
y(t) \, \frac{t}{(t-\mu-\hat\xi)} =1\ ,\quad \frac{(t-\mu-\hat\xi)\hat\xi}{(1-\hat\xi)(\hat\xi-w)}=1\ ,
\end{equation*}
where $y(t)$ is as in \eqref{defy}.
We find the solution
\begin{equation}\label{solhat}
\hat\xi(t)=t-\mu -t\, y(t)\ ,\quad w(t)= \frac{t(1-y(t))-\mu}{1+\mu -t(1-y(t))}\, (1+\mu-t)\ .
\end{equation}
As just mentioned, the contour of integration in $t$ must cross the real axis strictly inside the segment $[1+\mu-\rho,1+\mu]$ and on the negative real axis $(-\infty,0]$.
The saddle-point solution must have $t=(1-\hat\xi)(\hat\xi-w)/\hat\xi+\mu+\hat\xi>0$, as $\hat\xi\geq w$ 
(from the condition $\ell\geq p$), and $0\leq \hat\xi < 1$. The range of validity of \eqref{solhat} is therefore for $t\in [1+\mu-\rho,1+\mu]$.
The tangent line through the rescaled points $(1+\mu,1-\xi(t))$ and $(1+\mu+w(t),1)$ has the equation
\begin{equation}\label{newtan}
w(t)\, Y-\hat\xi(t)\, (X-t)=0 
\end{equation}
We may compare this result with that of Eqs.~\eqref{arcticone} and \eqref{arctictwo}.
Introducing the quantity $x(t)$ defined by \eqref{defy}, we may express
\begin{equation*}
\hat\xi(t)= t-\mu -\frac{t-1-\mu}{x(t)}, \quad w(t)=\hat\xi(t) \, \frac{x(t)}{x(t)-1}\ ,
\end{equation*}
which allows to identify the parametric representation \eqref{newtan} for the tangents with that \eqref{eq:tgeq} obtained in Section \ref{asymptoone},
or that \eqref{eq:tgeqtwo} obtained in Section \ref{asymptotwo}. 
We deduce that the arctic curve has the {\it same parametric expression} in terms of $t$ and $x(t)$
as before, and Theorem \ref{arcticthreethm} follows.
\end{proof}

To relate the function  $x(t)$ of Theorem \ref{arcticthreethm} to that given by \eqref{defx}, we use again the complementarity of the 
$a$'s and $b$'s which implies that 
$\sum_{i=0}^n\frac{1}{t-a_i}+\sum_{j=1}^m \frac{1}{t-b_j}=\sum_{i=0}^{n+m}\frac{1}{t-i}$. This 
leads immediately to
\begin{equation*}
e^{-\textstyle\int_0^1 \frac{du}{t-\al(u)}}\ e^{-\textstyle\int_0^\mu \frac{du}{t-\beta(u)}}=  \frac{t-1-\mu}{t}
\end{equation*}
which allows to identify the quantity $x(t)$ defined by \eqref{defy} to that defined by \eqref{defx} when both 
terms are well-defined (and positive), i.e.\ for $t\leq0$ or $t\geq \al(1)=1+\mu$. Eq.\eqref{defy} allows to extend the definition
of $x(t)$ to values of $t>\beta(\mu)=1+\mu-\rho$, i.e.\ to the new domain $[1+\mu-\rho,1+\mu]=[\al(1)-\rho,\al(1)]$. This 
corresponds to an analytic continuation of $x(t)$ in this
interval, leading to real values $x(t)\leq 0$ (from \eqref{defy}, as $y(t)>0$), a scheme which matches precisely 
that described in Section \ref{revisited} to extend the arctic curve for a freezing boundary creating a frozen domain of type R.
The analytic continuation of $x(t)$ may be obtained directly from the original definition \eqref{defx} of $x(t)$ which states that, for $t\geq \al(1)=1+\mu$,
\begin{equation}
x(t)=e^{\textstyle -\int_{0}^{1-\rho} \frac{du}{t-\al(u)} -\int_{1-\rho}^1 \frac{du}{t-\al(u)}} =-
\frac{1+\mu-t}{t-1-\mu+\rho}\, e^{-\textstyle \int_{0}^{1-\rho} \frac{du}{t-\al(u)} }
\label{newdefx}
\end{equation}
where we have used the freezing condition that $\al(u)=u+\mu$ on the segment $[1-\rho,1]$ (with $\mu=\al(1)-1$).
The last expression above allows to define $x(t)$ for $t\geq \al(1-\rho)=\al(1)-\rho=1+\mu-\rho$ and is equivalent to the definition \eqref{defy}
(as easily deduced from the identity $\sum_{i=0}^{n-r}\frac{1}{t-a_i}+\sum_{j=1}^m \frac{1}{t-b_j}=\sum_{i=0}^{n+m-r}\frac{1}{t-i}$). 
When $t$ increases from $1+\mu-\rho$ to $1+\mu$, $x(t)$ increases from $-\infty$ to $0$, or equivalently the slope $(x(t)-1)/x(t)$
of the tangent increases from $1$ to $+\infty$.

Let us examine the extremities of the new portion of arctic curve. For $t\to (1+\mu)^-$, writing $t=1+\mu-\epsilon$ in \eqref{newdefx} yields $x(t)\underset{\epsilon\to 0+}{\sim} -\epsilon\,  C$ with $C=\frac{1}{\rho} e^{-\int_0^{1-\rho}\frac{du}{1+\mu-\al(u)}}$, which yields 
\begin{equation*}
\begin{split}
&X(t)\underset{t\to (1+\mu)^-}{\to} 1+\mu =\al(1)=X_0\\
&Y(t)\underset{t\to (1+\mu)^-}{\to} \rho \ e^{\textstyle{\int_0^{1-\rho}\frac{du}{1+\mu-\al(u)}}}=e^{\textstyle{\int_0^{1-\rho}du \left\{\frac{1}{\al(1)-\al(u)}-\frac{1}{1-u}\right\}}}\\
& \qquad \qquad =e^{\textstyle{\int_0^{1}du \left\{\frac{1}{\al(1)-\al(u)}-\frac{1}{1-u}\right\}}} =Y_0
\end{split}
\end{equation*}
with $X_0$ and $Y_0$ as in \eqref{eq:X0Y0} (again we used $\al(1)-\al(u)=1-u$ for $u\in[1-\rho,1]$). The new portion of the arctic curve therefore connects to the previous known portion at $(X_0,Y_0)$.
For $t\to 1+\mu-\rho =\beta(\mu)$,  writing $t=\beta(\mu)+\eta$ and letting $\eta\to 0+$, we have $y(t)\simeq \eta^{1/\beta'(\mu)}$ (with some unimportant multiplicative constant), hence $x(t) \simeq \eta^{-1/\beta'(\mu)}$ and  $x'(t) \simeq \eta^{-1/\beta'(\mu)-1}$, leading to
$X(t)\to (1+\mu-\rho)$ and $Y(t)\to 0$ in the generic case $\beta'(\mu)>1$. The extremity of the new portion is thus at $(1+\mu-\rho,0)$, as expected.
\medskip

To summarize this section, the explicit computation above proves our conjecture of Section \ref{revisited} in the particular case where the freezing 
occurs on the right edge of the lower boundary of $D$. Clearly, the freezing of the left edge is amenable to the same exact calculation by 
a simple application of the reflection principle of Section \ref{reflection}, thus proving the conjecture in this case as well. 

\subsection{Examples}

\subsubsection{Fully frozen boundaries}
We display here examples where the boundary is fully frozen, namely where the distribution of starting points $a_i$
alternates between macroscopic portions with $a_{i+1}-a_i=1$ and macroscopic gaps with no $a$'s.
In turn, this corresponds to
piecewise linear $\al(u)$ with pieces corresponding exclusively of $p=1$ and $p=\infty$ portions. In general, 
we consider $2k-1$ positive numbers $\gamma_1,\gamma_2,\dots,\gamma_k,\delta_1,\delta_2,\dots,\delta_{k-1}$, 
together with $\delta_0=\gamma_0=0$ and such that
$\sum_{i=1}^k \gamma_i=1$. As before we introduce the quantities $\varphi_i:=\sum_{j=0}^{i} \gamma_i$, $i=0,1,\dots,k$,
with $\varphi_0=0$ and $\varphi_k=1$, as well as $\theta_{2i}=\sum_{j=0}^i (\gamma_j+\delta_j)$ and 
$\theta_{2i+1}=\theta_{2i}+\gamma_{i+1}$ for $i=0,1,\dots,k-1$.
We have for $i=1,2,\dots,k$:
$$ \al(u) = u+\sum_{j=1}^{i-1}\delta_{j} \qquad (u\in [\varphi_{i-1},\varphi_i))\ . $$
This immediately gives:
$$ x(t)= \prod_{j=1}^k \frac{t-\theta_{2j-1}}{t-\theta_{2j-2}}\ . $$

The simplest non-trivial example is for $k=2$. Let us denote $\gamma_1=a$, $\delta_1=b$ and $\gamma_2=c=1-a$.
The path problem is then equivalent (up to a simple shear/dilation) to that of the rhombus tiling of a hexagon 
with edge lengths $na,nb,nc$ (see Figure \ref{fig:ellipse}), and the arctic curve is
well known to be an ellipse. Noting that
\begin{equation*}
\left\{
\begin{split} 
X(t)&=(a+c)\frac{\big(a (a + b + c) - (a + c) t\big)^2 + b c t^2}{\big(a (a + b + c) - (a + c) t \big)^2+a b c\, (a+b+c)}\\
Y(t)&=(a+c)\frac{\big(a (a + b + c) - (a + c) t\big)^2}{\big(a (a + b + c) - (a + c) t\big)^2+a b c\, (a+b+c)} 
\end{split}
\right.
\end{equation*}
and eliminating $t$, we indeed find the equation of the arctic ellipse:
$$
\big((c-b) Y - (a + c) X+a(a+b+c)\big)^2+4 bc\, Y(Y-X)=0\ .
$$
\begin{figure}
\begin{center}
\includegraphics[width=13cm]{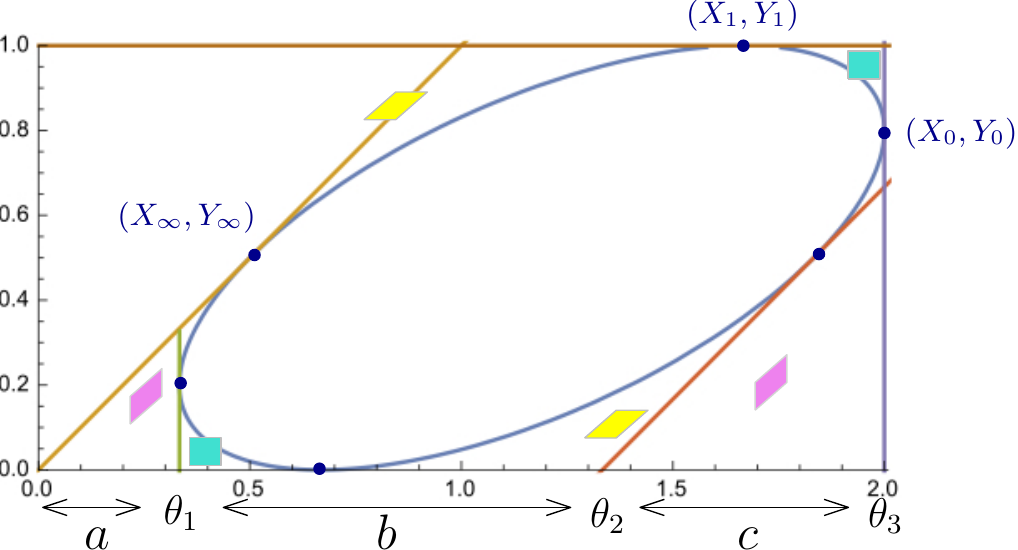}
\end{center}
\caption{The complete arctic curve when $\al(u)$ is a piecewise linear function made of two pieces of width $a$ and $c$ and slope $1$, with a discontinuity by $b$ in-between (here for $a=1/3$, $b=1$ and $c=2/3$). 
The pieces of slope $1$ give rise to freezing boundaries on the left and on the right of the lower boundary of $D$ and create frozen domains 
of type R, while the discontinuity gives rise to a central freezing boundary with a gap along the segment $[\theta_1,\theta_2]$, creating two frozen domains of respective type F and U. 
The resulting arctic curve is an ellipse, as expected since, up to a shear, the path/tiling problem is equivalent to that of the rhombus tiling of a hexagon with edge lengths $na,nb,nc$. }
\label{fig:ellipse}
\end{figure}
The case $a=\frac{1}{3}$, $b=1$, $c=\frac{2}{3}$ is represented in Figure \ref{fig:ellipse}.

\begin{figure}
\begin{center}
\includegraphics[width=15cm]{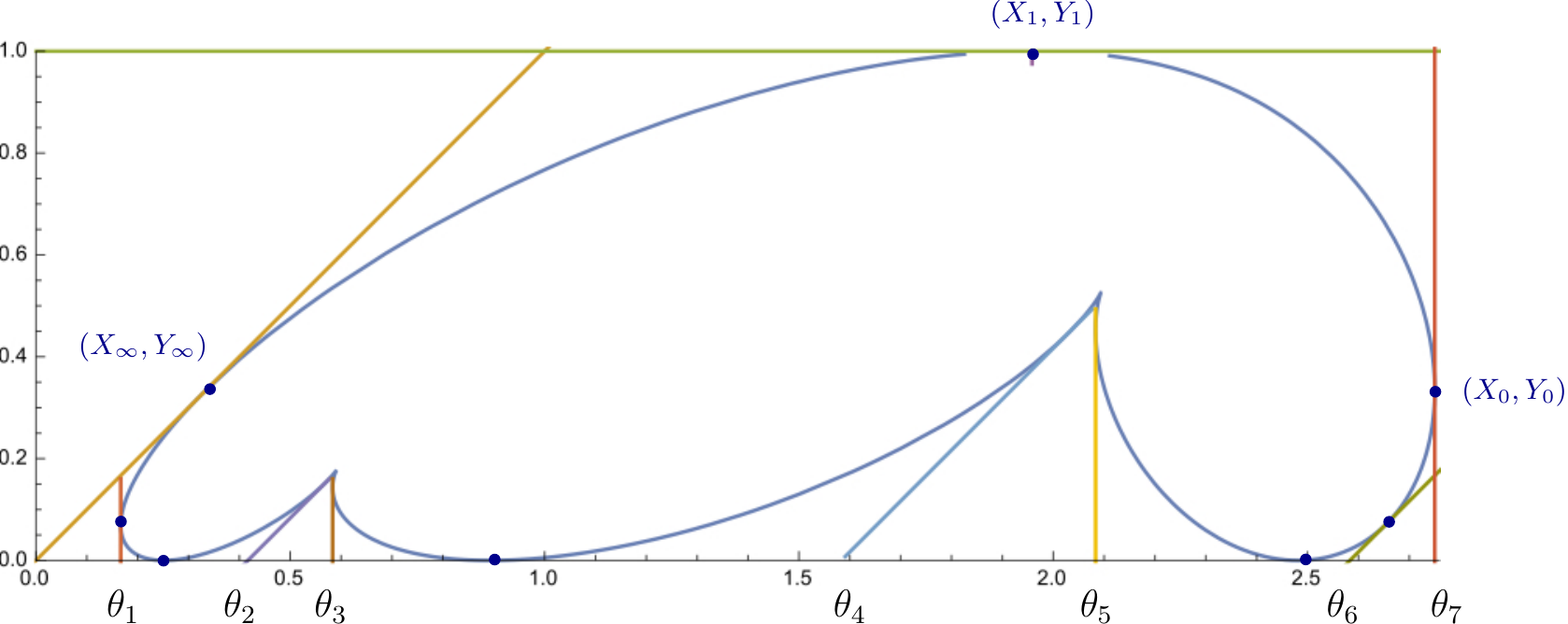}
\end{center}
\caption{The complete arctic curve when $\al(u)$ is a (discontinuous) piecewise linear function consisting of an alternation of pieces with slope $1$
of widths $\gamma_1=\gamma_2=\gamma_4=\frac{1}{6}$, $\gamma_3=\frac{1}{2}$ and of gaps $\delta_1=\frac{1}{4}$, $\delta_2=1$ and $\delta_3=\frac{1}{2}$. This results in an alternation of frozen domains adjacent to the lower boundary of the domain $D$.}
\label{fig:freeze}
\end{figure}
We display a more involved case with $k=4$ in Figure \ref{fig:freeze}, with $\gamma_1=\gamma_2=\gamma_4=\frac{1}{6}$, $\gamma_3=\frac{1}{2}$, 
$\delta_1=\frac{1}{4}$, $\delta_2=1$ and $\delta_3=\frac{1}{2}$, so that
$\theta_1=\frac{1}{6}$, $\theta_2=\frac{5}{12}$, $\theta_3=\frac{7}{12}$, $\theta_4=\frac{19}{12}$, $\theta_5=\frac{25}{12}$, $\theta_6=\frac{31}{12}$, $\theta_7=\frac{11}{4}$.

\subsubsection{Mixed boundaries}

We now consider a ``mixed" boundary case, with:
$$\al(u)=\left\{ \begin{split}
& u+u^2\ {\rm for}\ u\in \left[0,1/2\right]\\
& 1+u \ \ \ {\rm for} \ u\in \left(1/2,1\right]\ .
\end{split} \right. 
$$
This combines a non-linear distribution on $[0,\frac{3}{4}]$, a gap with no $a$'s on $(\frac{3}{4},\frac{3}{2}]$, and a frozen
boundary with $\al'(u)=1$ on $(\frac{3}{2},2]$.
The corresponding $x(t)$ reads
$$
x(t)=\frac{t-2}{t-\frac{3}{2}}\left(\frac{1-4t+\sqrt{1+4t}}{1-4t-\sqrt{1+4t}}\right)^{\frac{1}{\sqrt{1+4t}}}
$$
\begin{figure}
\begin{center}
\includegraphics[width=13cm]{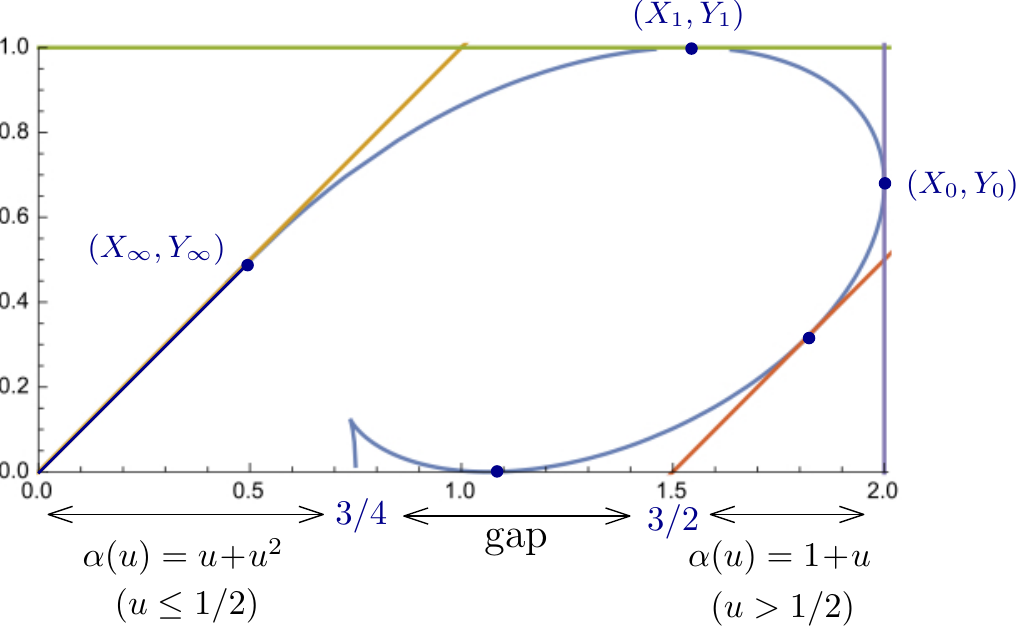}
\end{center}
\caption{The complete arctic curve when $\al(u)=u+u^2$ for $u\leq 1/2$ and $\al(u)=1+u$ for $u>1/2$, hence with a gap of linear size $3/2-3/4=3/4$.}
\label{fig:mixed}
\end{figure}
and the associated arctic curve is represented in Figure \ref{fig:mixed}.

\section{Conclusion}\label{conclusec}

\subsection{Summary and discussion}
In this paper, we have studied non-intersecting path models in the $\Z^2$ lattice with fixed arbitrary starting points 
along the $X$-axis. These fixed positions $a_0,a_1,...,a_n$ are described in the scaling limit $n\to\infty$
by a single piecewise differentiable increasing function $\al(u)$ with $\al'(u)\geq 1$ when defined, 
such that $a_i\sim n\al(i/n)$ for large $i,n$ with $i/n\to u$.  
Our main result \eqref{mainresult} is a parametric
expression \eqref{mainresult} of the arctic curve for the large $n$ asymptotic path model, involving some 
function $x(t)$ directly related to $\al(u)$ via \eqref{defx} (or its analytic continuation via \eqref{defy}). Several portions of the arctic curve are obtained from several 
intervals in the variable $t$. Explicit calculations were performed for three portions: two generic ones and one arising
in the presence of a freezing edge. We also analyzed, without explicit derivation, the shape of new portions induced by more
general freezing boundaries. 

It is interesting to better understand the meaning of the fundamental function $x(t)$.
First, we note that, associated to the asymptotic boundary ``shape"
is the actual \emph{distribution of starting points}, which can be defined in the finite size as:
$$\rho_n(v)=\frac{1}{n+1}\sum_{i=0}^n \delta(v-a_i)\ .$$
The limiting distribution is then defined on $[0,\al(1)]$ as
$$\rho(v)=\lim_{n\to\infty}n\,  \rho_n(nv)= \int_{0}^1 du\ \delta(v-\al(u))=\frac{1}{\al'(\al^{-1}(v))}\ ,$$
where $\al^{-1}(v)$ is the composition inverse of the function $\al(u)$ whenever well-defined. We may consequently interpret 
the quantity $x(t)$ of \eqref{defx} as giving the \emph{moment generating function (or resolvent)} of the distribution $\rho$,
namely:
\begin{equation}\label{momgen}
-{\rm Log}(x(t))=\sum_{n=0}^\infty \frac{\mu_n}{t^{n+1}}, 
\quad \mu_n=\int_0^{\al(1)} v^n\rho(v) dv=\int_0^1\al(u)^n du\ .
\end{equation}

Another remark is that the formula \eqref{mainresult} for the arctic curve may be rephrased in the language of the Legendre transformation as follows:
introducing the quantity
\begin{equation}
s(t):=\frac{x(t)}{1-x(t)}\ ,
\label{eq:soft}
\end{equation}
the equation \eqref{eq:tgeq} for the tangent line may be rewritten as 
\begin{equation*}
X=t-s(t)\, Y
\end{equation*} 
so that, if we express the arctic curve \eqref{mainresult} by its Cartesian equation $X=X(Y)$, the quantities $t$ and $s(t)$ are respectively the value at the origin ($Y=0$) 
and minus the slope of the line tangent to $X(Y)$ at the point $Y=Y(t)$. In particular, at $Y=Y(t)$, we have $s(t)=-X'(Y)$ (a relation which 
may also be checked directly from \eqref{mainresult}) and, inverting $s=s(t)$ into $t=t(s)$, we may write the above relation as 
\begin{equation*}
t(s)=X(Y(s))+s\, Y(s)\qquad \hbox{where} \ Y(s)=- X'^{-1}(s)
\end{equation*}
in terms of the composition inverse $X'^{-1}$ of the function $X'(Y)$.
This states that the function $t(s)$ is simply the Legendre transform of the function $X(Y)$ and vice versa, to that we may write as well
\begin{equation*}
X(Y)=t(s(Y))-Y\, s(Y)\qquad \hbox{where} \ s(Y)=t'^{-1}(Y)
\end{equation*}
in terms of the composition inverse $t'^{-1}$ of the function $t'(s)$.
This latter expression allows to directly get the location $X(Y)$ of the arctic curve as the Legendre transform of the function $t(s)$,
the composition inverse of $s(t)$ given by \eqref{eq:soft}. In practice, the equation $s=s(t)$ may have several solutions in $t$ so that
$t(s)$ can be made of several branches. Each branch gives in turn one branch for $X(Y)$ (recall that $X(Y)$ is made of at least two 
branches corresponding to the two generic portions of the arctic curve) or several ones with cusps if $t''(s)$ vanishes for some $s$.    

\medskip

We conclude this paper with three comments: we first give the equation for the arctic curve in modified coordinates adapted
to the rhombus tiling interpretation. We then discuss a direct geometric construction of the arctic curve inspired by the 
well-known Wulff construction for crystal shape. We end by a more technical point on some alternative use of the tangent method
consisting in moving the extremal starting point instead of the ending one.

\subsection{Rhombus tilings}
\begin{figure}
\begin{center}
\includegraphics[width=12cm]{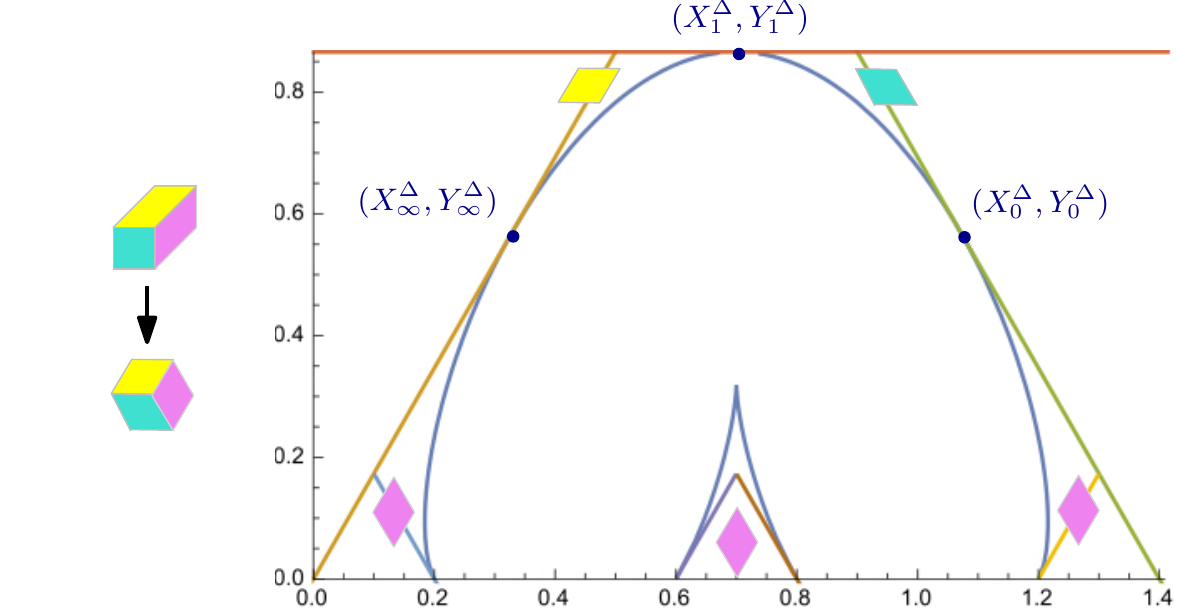}
\end{center}
\caption{The ``rectified" version of Figure \ref{fig:p1p2p1p2p1}, obtained by modifying the tiles as shown on the left. The vertical axial symmetry of the boundary condition induces a
vertical axial symmetry of the arctic curve.}
\label{fig:redressed}
\end{figure}
The problem we studied was conveniently expressed in terms of paths on the lattice $\Z^2$. However, the dual
tiling problem has the natural symmetry of the triangular lattice, the tiles being the three possible rhombi obtained by gluing pairs of adjacent triangles. All the results of this paper can be reformulated in this framework, provided we perform a change of
coordinates:
$$ (X,Y)\mapsto \left(X^{\Delta}=X-\frac{1}{2} Y, Y^{\Delta}=\frac{\sqrt{3}}{2}\, Y\right)\ . $$
Some of the symmetries observed in this paper become more manifest in this frame. 
For illustration we have represented in Figure \ref{fig:redressed} the ``rectified" version of the case of Figure \ref{fig:p1p2p1p2p1},
of a $k=5$ piecewise linear $\al(u)$, with a manifest vertical axial symmetry.

In the new coordinates, the arctic curve reads:
$$\left\{ \begin{split}
& X^{\Delta}(t)=t +\frac{x(t)^2-1}{2 \, x'(t)}  \\
& Y^{\Delta}(t)=\sqrt{3}\, \frac{(x(t)-1)^2}{2 \, x'(t)}\ .
\end{split}
\right.
$$
The corresponding parametric family of tangent lines has equation:
$$ (1+x(t))\, Y +\sqrt{3}(1-x(t))\, (X-t)=0\ .
$$

\subsection{A geometric construction}

One may wonder whether our result \eqref{mainresult} connecting the boundary conditions to the shape
of the arctic curve has a direct geometric description. It is very reminiscent indeed of the so-called Wulff
construction that relates the surface tension of a growing two-dimensional crystal to the shape of its boundaries. In that
case, the crystal is grown from an initial center $(x_0,y_0)$, with a surface tension $\sigma(\theta)$ depending on the angle $\theta$
measuring the orientation of the normal to the growing surface with respect to the microscopic crystalline axes. This surface tension
may be represented by the one-dimensional curve $r=\sigma(\theta)$ in polar coordinates centered at $(x_0,y_0)$: the shape of the boundary 
of the crystal is then (up to a global scaling) given by the envelope of the family of lines $L(\theta)$ that are normal to the radius vector
at the point $(\sigma(\theta),\theta)$ (more precisely the shape is given by the convex hull of this envelope).  

If we could interpret our family of tangent lines as arising from some Wulff construction, it would give access to
some candidate surface tension $\sigma(\theta)$. However, the problem is ill-posed, as there seems to be no
favored choice of the center $(x_0,y_0)$, and in fact if we were to think of our model as the final stage of some growth
process, it would rather start from frozen boundaries, and the status of fixed boundaries with arbitrary $\al(u)$ is unclear
in that respect.

\begin{figure}
\begin{center}
\includegraphics[width=14cm]{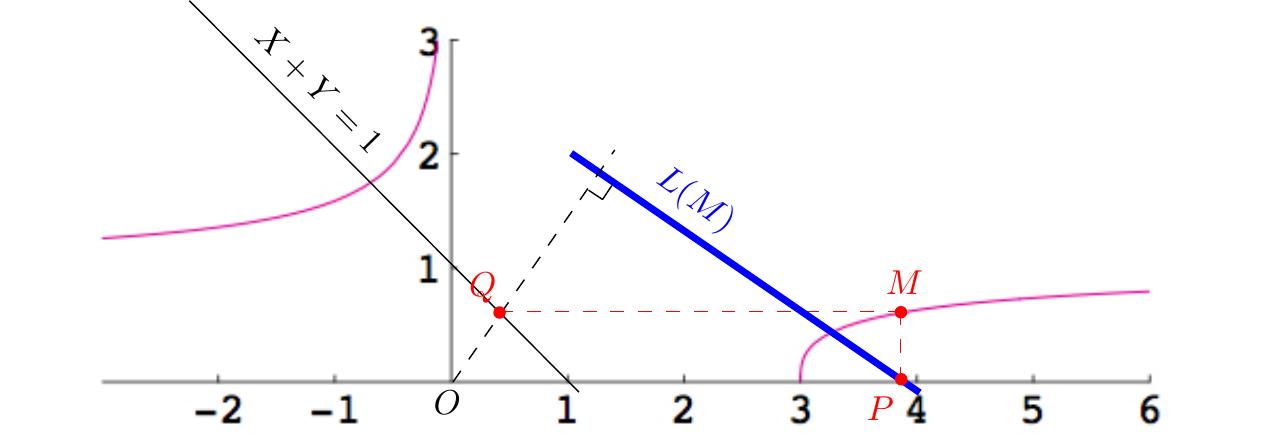}
\end{center}
\caption{Geometric construction of the line $L(M)$ from a point $M$ on the curve $(t,x(t))$ in some orthonormal frame. The line $L(M)$
(fat line) is the line orthogonal to $(OQ)$ passing through $P$. Moving the point $M$ along the curve generates a family of lines $L(M)$
whose envelope is the arctic curve.}
\label{fig:geometric}
\end{figure}
\medskip
On the other hand, we may devise the following direct geometrical construction for the arctic curve \eqref{mainresult} based again, 
in the spirit of the Wulff construction, on the data of some one-dimensional curve in the plane. Here this curve 
is simply the plot of the function $x(t)$ itself, namely the curve $(t,x(t))$ in cartesian coordinates (using some orthonormal basis). 
Given a point $M$ on this curve, we may easily obtain 
the corresponding value of $t$ by projecting the point vertically on the $X$-axis as the resulting point is $P=(t,0)$ by definition. 
The point $Q$ of coordinates $(1-x(t),x(t))$ is obtained by now projecting $M$ horizontally on the line 
of equation $X+Y=1$ (see Figure \ref{fig:geometric}). Denoting by $O=(0,0)$ the origin, the tangent to the arctic curve labelled by $t$ is, from its equation \eqref{eq:tgeq}, the line $L:=L(M)$ \emph{orthogonal} to the line $(OQ)$ and passing trough the point $P$. Each point $M$ of the plot gives rise to 
a line $L(M)$ and the arctic curve is the envelope of these lines.    

\subsection{Moving the starting point}
\begin{figure}
\begin{center}
\includegraphics[width=11cm]{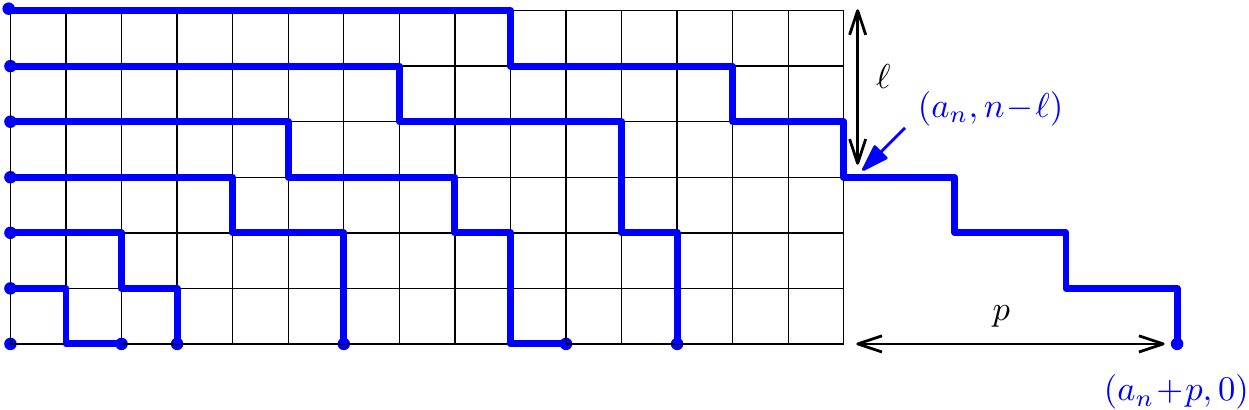}
\end{center}
\caption{The tangent method applied to the NILP with paths made of north- and west-oriented step, as obtained by moving the starting point 
of the outermost path is moved from $O_n=(a_n,0)$ to $O'_n=(a_n+p,0)$ 
with $p\in \Z_+$, forcing the path to re-enter the domain $D$ (here the displayed grid) by a west-oriented step at some position $(a_n,n-\ell)$ on the right boundary of $D$.}
\label{fig:Hnlcheck}
\end{figure}
So far we used the tangent method by moving the ending point of the outermost (or rightmost) path out of the domain $D$. 
Another choice would have been
to move instead the starting point of this path. Let us briefly describe how the method works in the original language of north- and west-oriented
paths. Moving the starting point $O_n=(a_n,0)$ to say $O'_n=(a_n+p,0)$ for some $p\in \Z_+$ forces the outermost path to re-enter the domain 
$D$ at some point $(a_n,n-\ell)$ on its right boundary (see Figure \ref{fig:Hnlcheck}). The partition function for NILP in the domain $D$ with
their outermost path starting at $(a_n,n-\ell)$, properly normalized by $Z_n$, defines our new one-point function $\check{H}_{n,\ell}$ 
for this new geometry. Its computation is made straightforward thanks to the remark that
\begin{equation*}
{\hat H}_{n,m,\ell}-{\hat H}_{n,m,\ell+1}=\check{H}_{n,\ell}-\check{H}_{n,\ell-1}
\end{equation*}
where $m=a_n-n$ and ${\hat H}_{n,m,\ell}$ as is \eqref{hatH}, which implies the sum rule
\begin{equation*}
\check{H}_{n,\ell}=1-{\hat H}_{n,m,\ell+1}\ .
\end{equation*}
These identities are obtained exactly via the same arguments as those given in Section \ref{combHHtile} to prove \eqref{eq:sumrule}.
Using the explicit expression \eqref{hatH} for ${\hat H}_{n,m,\ell}$, we may write
\begin{equation*}
1-{\hat H}_{n,m,\ell+1}=\frac{\displaystyle \prod_{s=1}^m (n+m-b_s)}{\displaystyle {n+m\choose n-\ell}}
\oint_{{\mathcal C}(b_1,b_2,\dots,b_m,n+m)} \frac{dt}{2{\rm i}\pi}
\frac{1}{(t-m-n)} \, \prod_{s=1}^m\frac{1}{(t-b_s)} \, \frac{\displaystyle \prod_{s=0}^{m+\ell-1}(t-s)}{\displaystyle  (m+\ell)!}
\end{equation*}
where the contour now encircles the pole at $m+n$, since, as easily checked,  its contribution produces the first term $1$ in the left hand side. 
Using now
\begin{equation*}
(t-m-n)\ \prod_{s=0}^{n-1}(t-a_s)\ \prod_{s=1}^m(t-b_s)=\prod_{s=0}^{n+m}(t-s)
\end{equation*}
and in particular, dividing by $(t-m-n)$ and setting $t=a_n=n+m$,
\begin{equation*}
\prod_{s=0}^{n-1}(a_n-a_s)\ \prod_{s=1}^m(n+m-b_s)=(n+m)!\ ,
\end{equation*}
the above expression yields immediately 
\begin{equation*}
\check{H}_{n,\ell}=\frac{1}{{\displaystyle \prod_{s=0}^{n-1} (a_n-a_s)}}
\oint_{{\mathcal C}(a_n-(n-\ell),a_n-(n-\ell-1),\dots,a_n)} \frac{dt}{2{\rm i}\pi}
\frac{1}{(t-a_n)} \, \prod_{s=0}^{n-1}(t-a_s) \, \frac{(n-\ell)!}{{\displaystyle \prod_{s=1}^{n-\ell}(t-a_n+s)}}\ .
\end{equation*}
It is then a straightforward exercise to use the tangent method machinery to get, in the large $n$ asymptotic regime,
the equation for the tangents and for the arctic curve. As expected, we recover the same set of tangents as in Section \ref{asymptoone},
given by equation \eqref{eq:tgeq} for $t\in [\al(1),+\infty)$. This provides an alternative derivation for the first portion of the arctic curve. 
Clearly, an alternative derivation for the second portion of arctic curve would consist in moving out of $D$ the starting point of the outermost path
for NILP configurations with paths made of east- and northeast-oriented steps.

\medskip

\noindent{\bf Acknowledgments.} We are thankful to Filippo Colomo, Christian Krattenthaler, Matthew F. Lapa, 
Vincent Pasquier and Andrea Sportiello for valuable discussions. 
PDF is partially supported by the Morris and Gertrude Fine endowment. EG acknowledges the support of the grant ANR-14-CE25-0014 (ANR GRAAL).

\bibliographystyle{amsalpha} 

\bibliography{arcticpaths}

\end{document}